\documentclass{article}

\usepackage{cmap}

\usepackage[T1]{fontenc}
\usepackage{lmodern}
\usepackage[utf8]{inputenc}

\PassOptionsToPackage{numbers}{natbib}
\usepackage[notheorems]{jmlrutils}
\usepackage{common_reqs}
\usepackage{common}

\usepackage{geometry}
\newcommand{\arxivversion}{} 

\title{Dynamic Algorithms for Online Multiple Testing}

\author{Ziyu Xu$^{1, 2}$ and Aaditya Ramdas$^{1, 2}$ \vspace{0.05in}\\
Machine Learning Department$^{1}$\\ Department of Statistics and Data Science$^{2}$\\ Carnegie Mellon University \\
\texttt{ziyux@cs.cmu.edu, aramdas@cmu.edu}}

\begin{document}
	\maketitle
	\begin{abstract}
We derive new algorithms for online multiple testing that  provably control false discovery exceedance (FDX) while achieving orders of magnitude more power than previous methods. This statistical advance is enabled by the development of new algorithmic ideas: earlier algorithms are more ``static''	 while our new ones allow for the dynamical adjustment of testing levels based on the amount of wealth the algorithm has accumulated. We demonstrate that our algorithms achieve higher power in a variety of synthetic experiments. We also prove that SupLORD can provide error control for both FDR and FDX, and controls FDR at stopping times. Stopping times are particularly important as they permit the experimenter to end the experiment arbitrarily early while maintaining desired control of the FDR. SupLORD is the first non-trivial algorithm, to our knowledge, that can control FDR at stopping times in the online setting.

	\end{abstract}
	\tableofcontents
\section{Introduction to Online Multiple Testing}
Online multiple hypothesis testing has become an area of recent interest due to the massively increasing rate at which analysts are able to test hypotheses, both in the sciences and in the tech industry. This has been a result of the growth in computational resources for automating the generation and testing of new hypotheses (and collection of data to test them). 

Online multiple testing is an abstraction of the scientific endeavor --- the broad goal is to test a (potentially infinite) sequence of hypotheses $H_1, H_2, \dots$ one by one over time. For example $H_i$ could state that \emph{Alzheimer's drug $i$ is no more effective than a placebo}, and if we reject this hypothesis, we are effectively proclaiming a ``discovery''. Some hypotheses are truly ``null'', meaning that there really is no effect of interest, and the rest are ``non-null'', but we do not know which are which, and we must use the data collected to identify as many non-null hypotheses as possible (making ``true discoveries''), while not rejecting too many truly null hypotheses (avoid making ``false discoveries''). Importantly, in the online setting, a decision to reject (or not reject) hypothesis $H_k$ at time $k$ is irrevocable and cannot be altered at a future time.

It is important to note that the setup assumes that algorithm (or analyst) never finds out whether a hypothesis was correctly rejected or not. In other words, we do not know (until perhaps years later, after a lot of followup time and money has been invested) which of our proclaimed discoveries were true discoveries and which were false. Thus the false discovery proportion (\(\FDP\)), which is the ratio of false to total discoveries is an unobserved random variable, and we would like to make many discoveries (have high ``power'') while keeping the \(\FDP\) small.


Previous algorithms for multiple testing have primarily focused on the false discovery rate (\(\FDR\)), which is the expectation of the \(\FDP\). However, in practice, users may be more interested in a probabilistic bound of the tail of the \(\FDP\) as it offers a stricter guarantee about the distribution of the \(\FDP\). Thus, we are interested in controlling the false discovery exceedance (\(\FDX\)) , which is the probability the \(\FDP\) exceeds a predefined threshold. 
To accomplish this, we first must describe how hypotheses are tested in the online multiple testing problem. As stated prior, we are testing a infinite sequence of hypotheses \(H_1, H_2, \dots\), and each hypothesis is either null or non-null. For each hypothesis, we perform some type of experiment, analyze the data, and arrive at a summary statistic known as a ``p-variable'' that is supported only in \([0, 1]\). Colloquially, p-variables are known as p-values, but we use the name ``p-variable'' to emphasize the fact that p-variables are random variables, and not a probability value. We use ``p-value'' to denote the specific value a p-variable is sampled at. Typically, p-variables are the result of making certain sampling assumptions about the data, and deriving some type of statistic with a known distribution from the data --- an excellent reference on how these statistics are derived is chapter 10 in \cite{wasserman_all_2004}. Hence, \(P_k\) is the p-variable whose distribution is based on the data drawn under hypothesis \(H_k\), and \(p_k\) will be the corresponding p-value to \(P_k\). If \(H_k\) is null, then the the distribution of \(P_k\) must be stochastically larger than uniform. We will formalize this notion later, but the key message here is that we essentially know what the distribution of the p-variable under the null hypothesis is.

As a result, an online multiple testing algorithm is primarily concerned with the stream of p-values \(p_1, p_2, \dots\), which correspond to the sequence of hypotheses being tested. For each hypothesis \(H_k\), the algorithm must output an alpha value, \(\alpha_k\), in \([0, 1]\), based solely on the previous p-values it has received i.e.\ \(p_1, \dots, p_{k - 1}\). Critically, it must output this alpha value before receiving the next p-value \(p_k\). \(H_k\) is then rejected if \(p_k\) is less than or equal to \(\alpha_k\). We denote the set of rejected hypotheses up to the \(k\)th hypothesis as \(\rejset_k \subseteq \left\{1, \dots, k \right\}\). The set of null hypothesis is denoted as \(\hypset_0\subseteq \naturals\).  We let the set of null rejections (i.e. false discoveries) in the first be \(k\) hypotheses be denoted as \(\rejnull_k = \hypset_0 \cap \rejset_k\). Thus, the \(\FDP\) of \(\rejset_k\) is defined as follows:
\begin{align*}
	\FDP(\rejset_k) &~\equiv~\frac{\left|\rejnull_k \right|}{\left|\rejset_k\right| \lor 1}.
\end{align*}
Define the (simultaneous, or time-uniform) \(\FDX_K\) as the probability that the \(\FDP\) ever exceeds a threshold \(\epsilon^* \in (0, 1)\) after time $K$:
\begin{align*}
	\FDX_K^{\epsilon^*}~\equiv~\prob{\underset{k \geq K}{\sup}\ \FDP(\rejset_k) \geq \epsilon^*}.
\end{align*} 
The \(\FDR\) in the online setting is a pointwise guarantee for each hypothesis i.e.\ it is the largest expected \(\FDP\) for any fixed time after \(K\):
\begin{align*}
	\FDR(\rejset_k) &~\equiv~ \expect\left[\FDP(\rejset_k)\right], \qquad  \FDR_K ~\equiv~\underset{k \geq K}{\sup}\ \FDR(\rejset_k).
\end{align*} (In prior work \(\FDR\) and \(\FDX\) have been defined with \(K = 1\).) We also provide the first algorithm that controls a related time-uniform error metric, which is the expectation of the largest \(\FDP\):
\begin{align*}
	\supFD_K \equiv \expect\left[\underset{k \geq K}{\sup }\ \FDP(\rejset_k)\right].   
\end{align*} Any algorithm that controls \(\supFD_1\) at a level \(\ell\) will also control \(\FDR_1\) at \(\ell\) as well, since \(\supFD\) is strictly larger than \(\FDR\). One central difference is that $\FDR_1 \leq \ell$ does not yield a guarantee at a data-dependent stopping time $\tau$, but $\supFD_1 \leq \ell$ does imply that $\FDR(\rejset_\tau) \leq \ell$. This yields the first FDR guarantee at stopping times in the literature. In past work ``anytime FDR control'' really held at any fixed time (when the time is specified in advance), but controlling $\supFD$ truly makes the aforementioned phrase actionable in the online setting, since it is a simultaneous guarantee over all time. We expand on the relationship between these two error metrics in \Cref{subsec:Contributions}.
We aim to maximize \textit{power}, which is the expected proportion of non-null hypotheses that are rejected:
\begin{align*}
	\Power(\rejset_k) ~\equiv~ \expect\left[\frac{\left|\rejset_k \cap \hypset_0^c\right|}{\left|\hypset_0^c\right|}\right].
\end{align*} 
Thus, the problem of online multiple testing is to design an algorithm that rejects as as many non-null hypotheses as possible while not rejecting null hypotheses to a degree that is sufficient to satisfy the desired level of control on an error metric (e.g. FDR, FDX, \(\supFD\)). 
\begin{quote}
We do not make guarantees about the optimality of power. This is because we do not have any information about the alternate hypothesis and make no assumptions about the distribution and frequency of non-null hypotheses. However, we can assert that the error metric is provably controlled, since it is only dependent on null hypotheses.
\end{quote}

The \(\LORD\) algorithm, introduced by \cite{javanmard2018online}, is an online multiple testing algorithm that provably controls the \(\FDR\) at level \(\ell\) (i.e.\  guarantees \(\FDR\) will be less than \(\ell\)). We will refer to \cite{javanmard2018online} as \(\JM\) in the sequel. \(\JM\) show that \(\LORD\), with additional constraints, controls \(\FDX\). We will refer to this version of \(\LORD\) as \(\LORD\FDX\). \textit{\(\LORD\FDX\) is the only existing rule that controls \(\FDX\) at a set level to the best of our knowledge.} We determine, through experiments, that \(\LORD\FDX\) is extremely conservative and we design an algorithm that significantly improves power while controlling \(\FDX\) at the same level.

Our proofs employ the post-hoc probabilistic bounds on the \(\FDP\) developed by \cite{katsevich2020}, which we refer to as \(\KR\) henceforth. \(\KR\) do not propose an algorithm for guaranteeing a certain level of \(\FDX\) control, but rather formulate an algorithm-agnostic bound based on \textit{any} sequence of alpha values $\{\alpha_i\}_{i\leq k}$ and rejection sets $\rejset_k$, in order to form a time-uniform confidence interval for the \(\FDP\). We view our work as an algorithmic mechanization of these bounds such that the \(\FDX\) can be controlled by the practitioner in a premeditated fashion.

\subsection{Notation and P-Variable Assumptions}

\paragraph{Abbreviations.} We abbreviate a few works that will be referenced often. As mentioned in the prequel, \JM\ refers to \cite{javanmard2018online} and \KR\ refers to \cite{katsevich2020}. \FS\ refers to \cite{foster2008alpha}, which is the original work in this area.

\paragraph{Notation.} Recall that \(P_k\) is the \(k\)th p-variable and \(\alpha_k\) is the \(k\)th alpha value. Let \(R_k = \ind{P_k \leq \alpha_k}\) indicate whether the \(k\)th hypothesis is rejected. Generally, script variables such as \(\hypset, \rejset, \filtration\) are sets. Capital letters, such as \(R, V, W, P\) are random variables. Greek letters, such as \(\alpha, \beta, \gamma, \delta, \epsilon\), represent user selected parameters for the algorithm. 
We collect the ``past rejection information'' into a filtration:
\begin{align}
    \filtration_k ~\equiv~ \sigma(\left\{R_i\right\}_{i \leq k}),
    \label{eqn:Filtration}
\end{align} with \(\filtration_0\) being the trivial \(\sigma\)-algebra. Note that \(\{\alpha_{k}\}_{k \in \naturals}\) is predictable with respect to \(\{\filtration_k\}_{k \in \naturals}\) meaning that $\alpha_{k}$ is \(\filtration_{k-1}\)-measurable.

\paragraph{P-variable assumptions.}In classical hypothesis testing, a p-variable is a random variable whose distribution is stochastically larger than uniform under the null, meaning that if the null hypothesis is true, we would have $\prob{P \leq s} \leq s$ for all $s \in [0,1]$. For the online setting, the weakest possible assumption one could make is that conditional on the past, the p-variables still satisfy such a condition.
This is usually formalized (by FS, \citealt{ramdas2017online}, and KR, for example), by the following \textit{conditional superuniformity} assumption about null p-variables:
\begin{align}
	\prob{P_{k + 1} \leq s \mid \filtration_k} \leq s \text{ for }k \in \hypset_0,\ s \in [0, 1].
	\label{eqn:SuperUniform}
\end{align} Note that the above assumption is weaker than assuming that the p-variables are independently, superuniformly distributed under the null hypothesis:
\begin{align}
    \prob{P_{k} \leq s } \leq s \text{ for }k \in \hypset_0,\ s \in [0, 1] \text{ and }\{P_k\}_{k \in \hypset_0} \text{ are jointly independent},
	\label{eqn:IndSuperUniform}
\end{align}
which is also commonly made in the online FDR literature (in \JM, for example). \revision{For the rest of the paper, we make no further assumptions on the relationship between on the hypotheses, \(H_1, H_2, \dots\) other than \eqref{eqn:SuperUniform}. We include assumption~\eqref{eqn:IndSuperUniform} because it implies \eqref{eqn:SuperUniform} and has been used to prove error control in prior work e.g.\ \(\LORD\) requires this assumption for \(\FDR\) control. Hence, the online multiple testing is a relatively assumption-light framework that can be adapted to many practical settings.}
Lastly, we define \(t_r\) to be the time step when the \(r\)th rejection is made, where \(r\) is a positive integer.

\subsection{Summary of Contributions}
\label{subsec:Contributions}

\begin{figure}[h]
	\centering
	\hspace*{-0.2in}
    \includegraphics[width=0.6\columnwidth]{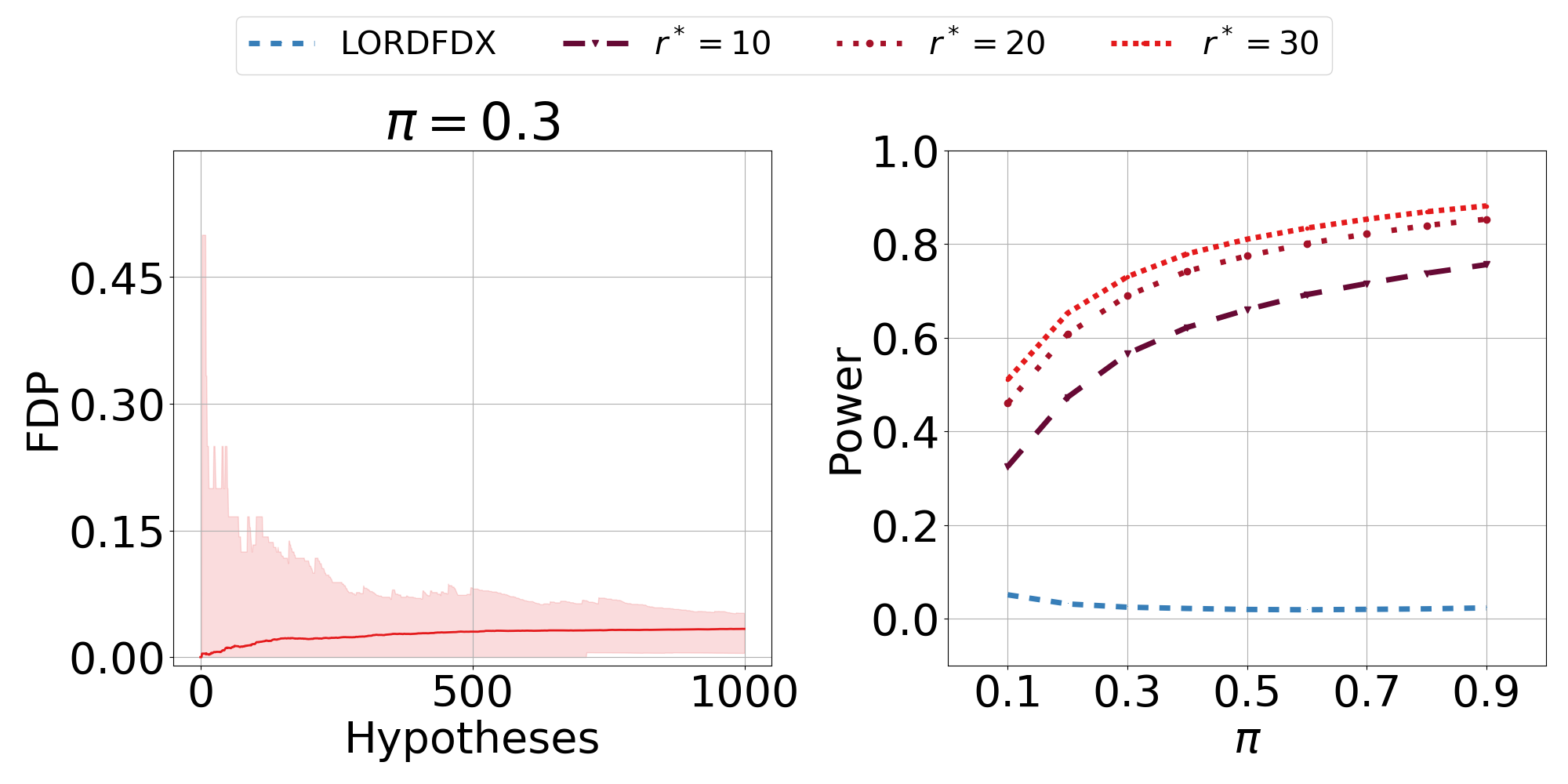}
	\caption{The right plot is the power of \(\SupLORD\), for different delays in rejections, \(r^*\), vs.\ \(\LORD\FDX\) from \(\JM\). The left plot is the average \(\FDP\) of \(\SupLORD\) with \(r^*= 30\) across hypotheses tested for non-null likelihoods \(\pi \in \{0.1, 0.3\}\). All algorithms control \(\FDX\) for \(\epsilon^*=0.15\) at a level of \(\delta^*=0.05\). Experiment details are in \Cref{subsec:FDXExperiments}. The shaded area is a \(1 - \delta^*\) confidence band for the \(\FDP\) --- it is smaller than \(\epsilon^*\) for the majority of hypotheses. This shows \(\FDX\) control at \(\epsilon^*\) holds for a substantial portion of the hypotheses. Also, the power of \(\SupLORD\) increases with \(r^*\). Thus, \(\SupLORD\) has significant \(\FDX\) control and more power than \(\LORD\FDX\) (by JM).}
	\label{fig:FDXPowerExample}
\end{figure}

\paragraph{Delayed FDX control.} We provide a formula for \(\SupLORD\), which controls \(\FDX\) at a set level, and has more power than \(\LORD\FDX\) (by JM). In fact, we observe empirically that \(\LORD\FDX\) rejects only a few hypotheses before becoming unable to reject any more hypotheses, and performs no better than the baseline approach we discuss in \Cref{subsec:AlphaSpending}.  We prove \(\SupLORD\) controls \(\FDX\) by using a time-uniform confidence interval introduced in \(\KR\) on the \(\FDP\). A novel feature of our algorithm is that we may choose the number of rejections after which we begin controlling \(\FDX\) in exchange for more power. Practically, users often test a large number of hypotheses, and make many discoveries. Thus, the small number of rejections with no \(\FDX\) control is an insignificant fraction of the experiment. 
The right plot in \Cref{fig:FDXPowerExample} shows that having the \(\FDX\) bound apply after only 10 rejections is sufficient for a noticeable increase in power over \(\LORD\FDX\). The left plot indicates that the \(\FDX\) is still controlled for a substantial proportion of the hypotheses, despite delaying control for 30 rejections. Consequently, we show through experimental results\footnote{Code for the figures and experiments in this paper can be found at \url{https://github.com/neilzxu/suplord}.} that by affording this flexibility to our algorithm, it can improve power with no practical deficit.

\paragraph{Dynamic scheduling.} Many multiple testing algorithms, including \(\SupLORD\), are based upon a notion of generalized alpha-investing introduced by \cite{aharoni2014generalized} which maintains a nonnegative value known as wealth. Wealth quantifies the error budget of the algorithm, and is an upper bound on the size of the next alpha value. 
For our error metrics of interest, an algorithm gains wealth by making rejections. We discuss and formalize the role of wealth in \Cref{sec:GAI}. Deciding how to allocate an algorithm's wealth to alpha values is a key challenge of designing online multiple testing algorithms. We devise a method for formulating alpha values that \textit{dynamically} chooses larger alpha values when wealth is large. Through experiments, we demonstrate that \(\SupLORD\) with dynamic scheduling outperforms spending schedules introduced in \JM, and fully utilizes its wealth and increases its power as a result.

\paragraph{Theoretical improvements in \(\FDR\) control.} The \(\SupLORD\) algorithm is able to control not only \(\FDX\), but also \(\FDR\) and \(\supFD\) at a set level of control. Table~\ref{table:TheoreticalComparison} illustrates the difference in assumptions and error control guarantees between \(\SupLORD\) and existing methods. Without imposing more restrictive independence assumptions and reducing the regime of control to fixed times, previous algorithms (\FS, \JM, \citealt{ ramdas2017online,ramdas2018saffron,tian2019addis}) could only control a modified \(\FDR\):
\begin{align}\label{def:mfdr}
    \mFDR(\rejset_k) &~\equiv~ \frac{\expect\left[|\rejnull_k|\right]}{\expect\left[\left|\rejset_k\right| \vee 1\right]}.
\end{align}
By controlling \(\supFD\), \(\SupLORD\) controls \(\FDR\) under the same guarantees and assumptions as previous algorithms that controlled \(\mFDR\). Specifically, \(\SupLORD\) offers the following theoretical improvements for control of \(\FDR\):
\begin{enumerate}
    \item \textit{Control at stopping times.}  Stopping times are random times that are a function of past events. Previous results have only controlled \(\FDR\) at fixed times. Consequently, this enables to the user to stop \(\SupLORD\) prematurely, and still maintain valid \(\FDR\) bounds.
    \item \textit{Non-monotonicity.} The spending schedule in prior work required \(\alpha_k\) to be \textbf{monotone}, that is:
    \ifarxiv{}{\vspace{-0.075in}}
    \begin{align}
        \alpha_k \text{ is coordinatewise nondecreasing w.r.t. past rejections } R_1, \dots, R_{k - 1} \ \text{for all }k \geq 1.
        \label{eqn:MonotoneSchedule}
    \end{align}
    \vspace{-0.3in}

    However, \(\SupLORD\) has a provable \(\FDR\) guarantee without this constraint on its spending schedule due to using techniques from \KR. Thus, when using non-monotonic schedules such as dynamic scheduling, \(\SupLORD\) maintains a provable \(\FDR\) guarantee.
    \item \textit{Non-independence}. Prior algorithms require independent p-values to have \(\FDR\) control. \(\SupLORD\) only requires a natural, weaker, baseline assumption we formalize in \eqref{eqn:SuperUniform}. This assumption is required to prove any known guarantee for online multiple testing algorithms.
\end{enumerate}
\ifarxiv{}{\vspace{-0.2in}}

\begin{table}[ht]
\caption{Comparison of the error metrics controlled under different assumptions of online multiple testing algorithms discussed in \Cref{subsec:Contributions}. Note that assumption~\eqref{eqn:IndSuperUniform} implies assumption~\eqref{eqn:SuperUniform}, and hence any guarantees achieved under~\eqref{eqn:SuperUniform} are also achieved under~\eqref{eqn:IndSuperUniform}. To achieve \(\FDR\) control, all algorithms except \(\SupLORD\) require both independent superuniform p-variables, and monotone spending schedules --- \(\SupLORD\) can offer \(\FDR\) guarantees under only conditional superuniformity \eqref{eqn:SuperUniform}. In addition, \(\SupLORD\) is the only algorithm that offers control of \(\FDR\) at stopping times and \(\FDX\) control without alpha-death, as \(\LORD\FDX\) suffers from alpha-death.}
\begin{center}
\begin{tabular}{|c|c|l|}
\hline
\textbf{Method} & \multicolumn{2}{c|}{\textbf{Guarantees}} \\ \hline
\(\LORD\), \(\mathrm{SAFFRON}\), \(\mathrm{ADDIS}\) &  & \multicolumn{1}{c|}{\textit{If p-variables satisfy}} \\
(\JM, \citealt{ramdas2017online} & \textit{If p-variables satisfy} & \multicolumn{1}{c|}{\textit{independence},} \\
\citealt{ramdas2018saffron}, & \textit{conditional superuniformity \eqref{eqn:SuperUniform}:} & \multicolumn{1}{c|}{\textit{superuniformity \eqref{eqn:IndSuperUniform}},} \\
\citealt{tian2019addis}) & \multicolumn{1}{l|}{\tabitem \(\mFDR\) at stopping times} & \textit{and monotone spending \eqref{eqn:MonotoneSchedule}:} \\
\multicolumn{1}{|l|}{} & \multicolumn{1}{l|}{} & \tabitem \(\FDR\) at fixed times \\
 &  & \tabitem \(\mFDR\) at stopping times \\ \hline
\multirow{4}{*}{$\SupLORD$} & \multicolumn{2}{c|}{\textit{If p-variables satisfy conditional superuniformity \eqref{eqn:SuperUniform}:}} \\
 & \multicolumn{2}{l|}{\tabitem \(\FDR\) at stopping times (Corollary~\ref{corollary:FDRStoppingTime})} \\
 & \multicolumn{2}{l|}{\tabitem \(\FDX\) w/o alpha-death (\Cref{thm:SupLORDFDX})} \\
 & \multicolumn{2}{l|}{\tabitem \(\mFDR\) at stopping times (\Cref{thm:mFDRSupLORD})} \\ \hline
\end{tabular}
\end{center}
\label{table:TheoreticalComparison}

\end{table}

\section{Generalized Alpha-Investing (GAI)}
\label{sec:GAI}

One of the key paradigms in online hypothesis testing is notion of alpha-investing that was first introduced by \cite{foster2008alpha}, and expanded upon in later works \citep{aharoni2014generalized, ramdas2017online}. 
Our primary error metric of interest in this section will be \(\FDR\), since prior algorithms, including \(\LORD\), have focused on controlling \(\FDR\). However, this paradigm is generalizable to any of the error metrics we consider in this paper. In particular, we will see how we use this idea to derive our new algorithm, \(\SupLORD\), in \Cref{sec:SupLORD}.

\subsection{Alpha-Spending and Alpha-Death}
\label{subsec:AlphaSpending}
To elucidate alpha-investing, we first discuss the notion of alpha-spending,
which is an Bonferroni correction. We control the FDR (indeed the familywise error rate, a more stringent criterion) by ensuring that the following constraint is satisfied at every \(k\):
\begin{align}
    \sum_{i = 1}^k\alpha_k \leq \ell.
    \label{eqn:BonferroniBound}
\end{align} In condition~\eqref{eqn:BonferroniBound}, we can view \(\ell\) as our error budget, or \textbf{wealth}, since the sum of all alpha values adds up to \(\ell\). Thus, we can imagine this as starting out with wealth of \(\ell\), and spending no more than our current wealth on the current alpha value for each hypothesis. This method for choosing alpha values is known as alpha-spending. We can explicitly define the wealth at time \(k\) for alpha-spending to be:
\begin{align}
    W(k) ~\equiv~ \ell - \sum_{i = 1}^k\alpha_i.
    \label{eqn:WealthFormulaBonferroni}
\end{align} 
Definition~\eqref{eqn:WealthFormulaBonferroni} indicates that wealth is monotonically decreasing when alpha-spending. As more hypotheses are tested, the alpha values will shrink towards zero. Thus, the algorithm will fail to reject nearly every future hypothesis after enough hypotheses are tested. \cite{ramdas2017online} term the shrinkage or disappearance of alpha values and consequent failure to reject hypotheses after a point in time as \textbf{alpha-death}. Alpha-spending suffers from alpha-death because the bound is extremely conservative --- Bonferroni controls the probability of making even a single false rejection to be at level \(\ell\) as opposed to simply controlling the \(\FDR\) to be at \(\ell\).

\subsection{Generalized Alpha-Investing Algorithms}
\paragraph{An alpha-investing example: \(\LORD\).} To avoid alpha-death, we can choose a less conservative bound than Bonferroni while maintaining \(\FDR\) control. \cite{ramdas2017online} provide an estimator that upper bounds the \(\FDR\):
\begin{align}
    \widehat{\FDP}_{\LORD}(\rejset_k) = \frac{\sum_{i = 1}^k \alpha_i}{\left|\rejset_k\right| \vee 1}.
    \label{eqn:FDPEstimator}
 \end{align} Notably, this is also the estimator that \(\LORD\) uses to control \(\FDR\). Thus, we simply need to control \(\widehat{\FDP}_{\LORD}\) at level \(\ell\), for the \(\FDR\) to be controlled at the same level. Hence, we get a different wealth formulation than in the alpha-spending case:
 \begin{align}
     W(k) ~\equiv~ \ell(|\rejset_k|  \vee 1) - \sum_{i = 1}^k \alpha_i.
     \label{eqn:FDPEstimatorWealth}
 \end{align} The key difference between this formulation and alpha-spending is that wealth can be earned. Thus, alpha-death is avoided since $\alpha_i$ need not shrink to zero. We obtain a reward of \(\ell\) each time the algorithm makes a rejection\footnote{In this formulation, the initial wealth is $\ell$ and no wealth is gained from the first rejection, but this can be easily amended so that the sum of the initial wealth and first reward equals $\ell$ using any other split.} i.e.\ increases the cardinality of \(\rejset_k\).
 The wealth reward follows from \eqref{eqn:FDPEstimator} --- when more hypotheses are rejected, the denominator in \eqref{eqn:FDPEstimator} increases, which allows for a larger numerator while maintaining the same \(\widehat{\FDP}_{\LORD}\) value. Thus, an algorithm may carefully ``invest'' its wealth on alpha values in a way that maximizes the number of rejections.
 
\paragraph{Generalized alpha-investing (GAI).}  GAI algorithms have the following wealth update:
 \begin{align}
     W(k) &~\equiv~ \beta_0 + \sum_{i = 1}^{k} \beta_iR_i  - \sum_{i = 1}^k\alpha_i,
     \label{eqn:AlphaInvestingWealth}
\end{align}
 where \(\{\beta_i\}_{ i \in \naturals \cup \{0\}}\) is the boost sequence and is nonnegative. \(\beta_i\) is the boost in wealth that is earned from the algorithm making the rejection on the \(i\)th hypothesis. Similar to alpha values, we require \(\{\beta_k\}_{k \in \naturals}\) to be predictable with respect to \(\{\filtration_k\}_{k \in \naturals}\) and \(\beta_0\) to be a constant. We can recover \eqref{eqn:WealthFormulaBonferroni} by setting \(\beta_i = 0\), and we can recover \eqref{eqn:FDPEstimatorWealth} by setting \(\beta_i = \ell\). The key invariant that all alpha-investing algorithms maintain is the following:
 \begin{align}
     \alpha_k \leq W(k - 1).
     \label{eqn:WealthInvariant}
 \end{align} Thus, selection of alpha values respects this notion of wealth and ensures it is always nonnegative.
 
\begin{figure}[h]
    \centering
    \includegraphics[width=0.8\textwidth]{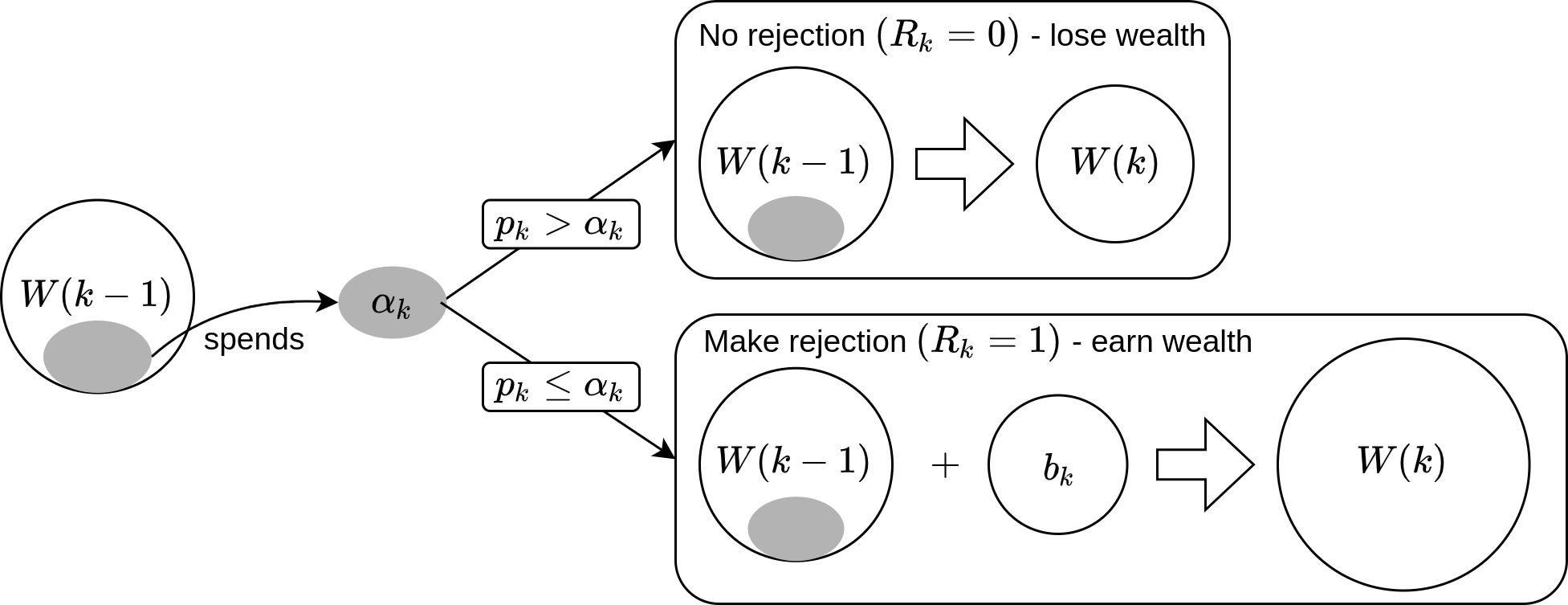}
    \caption{Diagram of the alpha-investing paradigm. The algorithm spends wealth at each hypothesis \(k\) and earns wealth if it makes a rejection, and loses \(\alpha_k\) wealth if it does not.}
    \label{fig:AlphaInvestingDiagram}
\end{figure}
 
\paragraph{Spending schedule.} \JM\ also proposed two methods of choosing alpha values (i.e. spending schedules) based on wealth. Let \(\{\spendcoef_i\}_{i \in \naturals}\) be a \textbf{spending sequence} that has the following properties:
\begin{subequations}
\begin{align}
    \spendcoef_i \geq 0 \text{ for all }i \in \naturals\\
    \sum\limits_{i = 1}^\infty \spendcoef_i = 1.
\end{align}
\label{eqn:SpendingConditions}
\end{subequations} \JM\ propose the following spending schedules:
\begin{align}
    \alpha_k = \begin{cases}
    \gamma_k \beta_0 + \sum\limits_{j = 1}^{\left|\rejset_{k - 1}\right|}\gamma_{k - t_j}\beta_{t_j} & \steady\\
    \gamma_{k - t_{|\rejset_{k - 1}|}}W(t_{|\rejset_{k - 1}|}) & \aggressive.
    \end{cases}
    \label{eqn:StaticSpendingSchedules}
\end{align} 
While both are valid spending schedules in the sense that neither would violate condition~\eqref{eqn:WealthInvariant}, only \steady\ can be used to provide provable \(\FDR\) control in \(\LORD\), since \steady\ is the only monotone rule \eqref{eqn:MonotoneSchedule}. Intuitively, this means that when more rejections are made, a monotone spending schedule should output larger alpha values at all future time steps than it would if fewer rejections were made. \JM\ show that \steady\ satisfies monotonicity while \aggressive\ does not. Hence, \aggressive\ does not provide any \(\FDR\) guarantees when used with \(\LORD\). In \Cref{sec:ControllingError}, we will discuss how \(\SupLORD\) can maintain \(\FDR\) control even when using non-monotone spending schedules. 

For the spending sequence, a default choice is to set \(\gamma_i \propto \frac{\log(i \vee 2)}{i\exp(\sqrt{\log i})}\), which is motivated by optimizing a lower bound for the power of \(\LORD\) in the Gaussian setting in \JM.
 
\section{\(\SupLORD\): Delayed \(\FDX\) Control}
\label{sec:SupLORD}

We propose a new alpha-investing algorithm, \(\SupLORD\), which provides control over the \(\FDX\). \(\SupLORD\) requires the user to specify \(r^*\), a threshold of rejections after which the error control begins to apply, \(\epsilon^*\), the upper bound on the \(\FDP\), and \(\delta^*\), the probability at which the \(\FDP\) exceeds \(\epsilon^*\) at any time step after making \(r^*\) rejections. We derive this algorithm from a new estimator, \(\overline{\FDP}\), that gives probabilistic, time-uniform control over the \(\FDP\). \revision{Notably, \(\SupLORD\) requires no assumptions about relationships between hypotheses tested except for condition~\eqref{eqn:SuperUniform} on p-variables of null hypotheses}.

\subsection{Deriving the Boost Sequence}
To prepare for the \(\overline{\FDP}\), define the following quantities, respectively:
\ifarxiv{
\begin{align*}
	\widehat{V}(\rejset_k) &~\equiv~ \sum\limits_{i = 1}^k \alpha_i,\\
	\widehat{\FDP}_{\SupLORD}(\rejset_k) &~\equiv~ \frac{\widehat{V}(\rejset_k) + 1}{\left|\rejset_k\right|}. 
\end{align*}
}{
\begin{align*}
	\widehat{V}_k ~\equiv~ \sum\limits_{i = 1}^k \alpha_i ~ , \quad \quad \widehat{\FDP}(\rejset_k) ~\equiv~ \frac{\widehat{V}_k + 1}{\left|\rejset_k\right|}. 
\end{align*}
}
Now, for any \(\delta \in (0, 1)\), define 
\begin{align*}
\coef{\frac{1}{\delta}} &~\equiv~ \frac{\log(\frac{1}{\delta})}{ \log\left(1 + \log(\frac{1}{\delta})\right)}.
\end{align*}
Hence our estimator of interest is defined as follows:
\begin{align*}
\overline{\FDP}(\rejset_k) &~\equiv~ \coef{\frac{1}{\delta^*}} \cdot  \widehat{\FDP}(\rejset_k).
\end{align*}
\vspace{-0.2in}
\begin{fact}[Theorem 4 from \(\KR\)]
Let \(\{\alpha_k\}_{k \in \naturals}\) be any sequence of alpha values predictable with respect to filtration \(\filtration_{k}\) in definition~\eqref{eqn:Filtration}.
Assuming conditional superuniformity \eqref{eqn:SuperUniform}, the following uniform bound holds:
	\begin{align*}
	\prob{\FDP(\rejset_k) \leq \overline{\FDP}(\rejset_k) \text{ for all }k \geq 1} ~\geq~ 1 - \delta.
	\end{align*}
	\label{fact:PosthocBound}
\end{fact}
\ifarxiv{}{\vspace{-0.3in}}
Thus, to control \(\FDX\), \(\SupLORD\) must ensure \(\overline{\FDP}(\rejset_k) \leq \epsilon^*\) for all \(k \geq t_{r^*}\). This is equivalent to requiring the following two conditions on \(\{\beta_j\}_{j \in \naturals \cup \{0\}}\):
\begin{subequations}
\begin{align}
    \sum\limits_{i = 0}^{t_{r^* - 1}} \beta_iR_i &\leq \frac{\epsilon^*r^*}{\coef{\frac{1}{\delta^*}}} - 1 \label{eqn:SumSupLordBoostCondition}\ \text{for}\ i \leq t_{r^* - 1}\\
    \beta_i & \leq \frac{\epsilon^*}{\coef{\frac{1}{\delta^*}}}\ \text{for all }i > t_{r^* - 1}. \label{eqn:IncreaseSupLordBoostCondition}
\end{align}
\label{eqn:SupLORDBoostConditions}
\end{subequations}
Unlike previously discussed alpha-investing algorithms, \(\SupLORD\) has two different conditions for its boost sequence. This is the result of \(\SupLORD\)'s capability to delay control of \(\FDX\) up to \(r^*\) rejections. We can think of \(\SupLORD\) as operating in two phases. The first phase of \(\SupLORD\) is when it has not made \(r^* - 1\) rejections yet. The constraint on this phase is formulated in condition~\eqref{eqn:SumSupLordBoostCondition}. In this phase, we are not concerned with how much boost in wealth the algorithm receives at each rejection in the first \(r^* - 1\) rejections. Thus, we only restrict the \textit{total wealth} \(\SupLORD\) attains in the first \(r^* - 1\) rejections so that it does not attain so much wealth that it can violate the bound on \(\overline{\FDP}\) when the \(r^*\)th rejection occurs. 

The second phase, when we have made \(r^*\) rejections, is an analog to the alpha-investing schemes of \(\LORD\) and alpha-spending, except for the estimator \(\overline{\FDP}\). Condition~\eqref{eqn:IncreaseSupLordBoostCondition} restricts the wealth boost for each rejection to be limited so that \(\overline{\FDP}\) is controlled to be less than \(\epsilon^*\). As a result, we can prove the following lemma:
\ifarxiv{}{\vspace{-0.05in}}
\begin{lemma}
    Any GAI algorithm with a boost sequence \(\{\beta_i\}_{i \in \naturals \cup \{0\}}\) that satisfies conditions~\eqref{eqn:SupLORDBoostConditions} will ensure
        \(\overline{\FDP}(\rejset_k) \leq \epsilon^*\)
    for all times \(k \geq t_{r^*}\) such that a rejection is made at time \(k\).
    \label{lemma:FDPBarLemma}
\end{lemma}
\ifarxiv{}{\vspace{-0.05in}}
A formal proof this statement is shown in \Cref{subsec:BoundedFDPProof}. Thus, by combining Lemma~\ref{lemma:FDPBarLemma} with \Cref{fact:PosthocBound}, we can derive the following guarantee for \(\SupLORD\).

\begin{theorem}
	Assuming conditional superuniformity \eqref{eqn:SuperUniform}, let \(0 < \epsilon^*, \delta^* < 1\), \(r^* \geq 1\) be a positive integer, and \(t_{r^*}\) be a random variable that is the time when the \(r^*\)th rejection is made. 
    A GAI algorithm with a boost sequence \(\{\beta_i\}_{i \in \naturals \cup \{0\}}\) that satisfies conditions~\eqref{eqn:SupLORDBoostConditions} guarantees:
    \begin{align*}
        \FDX_{t_r^*}^{\epsilon^*} \leq \delta^*.
    \end{align*}
    \label{thm:SupLORDFDX}
\end{theorem}
\ifarxiv{}{\vspace{-0.4in}}
A default choice for selecting a boost sequence that satisfies conditions~\eqref{eqn:SupLORDBoostConditions} is as follows:
\begin{align}
    \beta_i = \begin{cases}
     \frac{\frac{\epsilon^*r^*}{\coef{\frac{1}{\delta^*}}} - 1}{r^*}  &  i \leq t_{r^* - 1},\\ \frac{\epsilon^*}{\coef{\frac{1}{\delta^*}}} & i > t_{r^* - 1}.
    \end{cases}
    \label{eqn:DefaultBoostSupLORD}
\end{align} This boost sequence is tight with both upper bounds in conditions~\eqref{eqn:SupLORDBoostConditions}, and it evenly distributes the upper bound in condition~\eqref{eqn:SumSupLordBoostCondition} across the initial wealth and first \(r^* - 1\) rejections.

\subsection{Numerical Simulations and Real Data Experiments for \(\FDX\) Control}
\label{subsec:FDXExperiments}

\revision{We primarily explore the power of our procedures through simulations. The reason for this is simple: when working with real data, we do not know which of our discoveries are true and which are false (otherwise, scientific hypothesis testing would be very easy; we would simply avoid the false ones).
Definitively determining which discoveries are true in real data is generally intractable 
and requires extensive effort (e.g., after a large-scale genome-wide association study, if we proclaim that we have discovered a gene that regulates our propensity for some disease, it may take many years of followup targeted `gene knockout' experiments to determine if our discovery was genuinely true, or a false alarm). In contrast, with synthetic hypotheses and data, we know which hypotheses are null by construction, and can evaluate the power of our algorithms and confirm they do satisfy control of the desired error metric. Nevertheless, we do compare our methods on a real world dataset, the International Mouse Phenotype Consortium (IMPC) dataset that was created by \cite{karp_prevalence_2017} and has been used to test other online multiple testing algorithms in prior work \citep{tian_online_2021}.  However, because of the aforementioned issues, we can only count the total number of discoveries as a proxy for each algorithm's power (proportional to the total number of \emph{true} discoveries).}

We perform experiments for several baseline methods that control false discoveries. 
\begin{enumerate}
    \item
\textbf{LORD and LORDFDX.} We use the version of \(\LORD\) and \(\LORD\FDX\) with a \steady\ spending schedule from \cite{ramdas2017online} that is more powerful than its original formulation in \(\JM\), and provably controls \(\FDR\) to be under level \(\ell\).
\item
\textbf{Alpha-spending.} This is simply the online Bonferroni rule that controls the probability of any false discovery to be under level \(\ell\). We also use the \steady\ schedule for spending.
\end{enumerate}

We compare \(\SupLORD\) against existing methods for false discovery control.  We choose \(\delta^*=0.05,\ \epsilon^*=0.15\) for \(\SupLORD\) and \(\LORD\FDX\). In \(\SupLORD\), we choose \(r^* = 30\) for our simulations \revision{and \(r^*=100\) for the IMPC dataset, since it contains many more hypotheses than our simulations --- we also use the default boost sequence in \eqref{eqn:DefaultBoostSupLORD}}. We set \(\ell=0.05\) for \(\LORD\) and alpha-spending. For \(\LORD\) and \(\LORD\FDX\), we choose \(\beta_0 = 0.1\ell,\ \beta_1=0.9\ell\) and \(\beta_i = \ell\) for all \(i > 1\). All methods use the default \(\{\spendcoef_i\}_{i \in \naturals}\) where \(\gamma_i \propto \frac{\log(i \vee 2)}{i\exp(\sqrt{\log i})}\) as stated in \Cref{sec:GAI}.

\paragraph{Synthetic generation of p-values.} Our simulations follow a similar setup to \cite{ramdas2018saffron}. We compute p-values based on a one sided test for Gaussian values i.e.\ we sample a Gaussian \(Z_i\) and compute one-sided p-values for whether \(\mu_i > 0\) using the statistic \(P_i = \Phi(-Z_i)\) where \(\Phi\) is the standard Gaussian c.d.f. Let \(\pi_i\) be the probability of the \(i\)th hypothesis being non-null. We generate our \(Z_i\) as follows:
\begin{align*}
	B_i \sim \Bern(\pi_i), \quad \quad Z_i \sim \begin{cases}
		\Gaussian(0, 1) & B_i = 0\\
		\Gaussian(\mu_i, 1) & B_i = 1.
	\end{cases}
\end{align*}

Here \(B_i\) is a Bernoulli variable that denotes whether the \(i\)th hypothesis is null. If \(B_i = 0\), then we sample \(Z_i\) from the null distribution, which is the standard Gaussian. If \(B_i = 1\), then we draw \(Z_i\) from the non-null distribution, which is a shifted standard Gaussian with mean \(\mu_i\). We may select different signal strengths \(\mu_i\), and non-null likelihoods \(\pi_i\) at each time step for our experiments. For our primary simulation, we choose constant values \(\mu, \pi\) such that \(\mu_i = \mu, \pi_i = \pi\) across all \(i\). Addition simulations with other ways of selecting \(\mu_i\) and \(\pi_i\) are discussed in \Cref{sec:AdditionalExperiments}.

Ultimately, we generate 200 experiments with 1000 hypotheses for each data setting. Standard errors are negligible in all of our numerical results, and we do not plot them.

\paragraph{Real data.} \revision{The IMPC dataset consists of a series of experiments that aims to identify which mouse phenotypes each mouse protein coding gene has an effect on. \cite{karp_prevalence_2017} analyze the raw experimental data, and produce a series of p-values corresponding to each pairwise combination of gene and phenotype. Thus, the null hypothesis of each experiment that knocking out a single gene \textit{i} has no effect on a single phenotype of interest \textit{j}. As a result, there are 172328 total hypotheses in the dataset. This method of experimentation naturally fits an online structure, because scientists may develop new phenotypes to test against or discover new genes to analyze over time, and consequently want be able to freely extend the set of hypotheses they are testing.}


\paragraph{Results for different alpha-investing algorithms.} In \Cref{fig:BaselineGaussian}, we observe that the only other online \(\FDX\) control algorithm, \(\LORD\FDX\), has significantly less power than \(\SupLORD\). In fact, \(\LORD\FDX\) is on par or worse than the Bonferroni bound of alpha-spending. Thus, \(\SupLORD\) is able to obtain more power by delaying control for a constant number of rejections. Similar results are shown in the HMM setting, which we defer, along with additional metrics, to \Cref{sec:AdditionalExperiments}. \revision{Further, these results are consistent with \(\SupLORD\) rejecting the most hypotheses on the IMPC data set in \Cref{table:RealDataFDX} --- the gains in power seem to translate to a real world setting.}

\begin{wraptable}{r}{0.3\linewidth}
    \centering
    \caption{\label{table:RealDataFDX}\revision{Comparing total discoveries (rejected hypotheses) by each algorithm on the IMPC data. \(\SupLORD\) rejects the most hypotheses.}}
    \begin{tabular}{l|r}
         \revision{\textbf{Method}} & \revision{\textbf{Rejections}}\\
         \hline
         \revision{Bonferroni} & \revision{879} \\
         \revision{\(\LORD\)} & \revision{12352} \\
         \revision{\(\LORD\FDX\)} & \revision{8} \\
         \revision{\(\SupLORD\)} & \revision{\textit{18793}}
    \end{tabular}
\end{wraptable}

Interestingly, these results also suggest a reversal of the previous notion that \(\FDX\) control was stricter than \(\FDR\) control. Previously, \(\LORD\FDX\) for controlling \(\FDX\) at some \(\delta^*\) was demonstrated by \(\JM\) to be less powerful than \(\LORD\) controlling \(\FDR\) at a level of \(\ell = \delta^*\) in experiments. This is due to \(\LORD\FDX\) being \(\LORD\) with an additional constraint i.e.\ if this additional constraint is violated, \(\LORD\FDX\) enters alpha-death. Hence, it is impossible for \(\LORD\FDX\) to make more rejections that \(\LORD\) when they are controlled at equivalent levels i.e.\ \(\ell = \delta^*\). On the other hand, \(\SupLORD\) uses the time-uniform bounds from \KR\ to establish control of \(\FDX\). \(\SupLORD\) suggests that controlling \(\FDX\) is not necessarily more restrictive than controlling \(\FDR\). Since high probability bounds may be of greater interest to the user in practice, \(\SupLORD\) demonstrates that controlling \(\FDX\) does not necessarily  sacrifice power.

\begin{figure}[h!]
    \centering
	\includegraphics[width=0.7\textwidth]{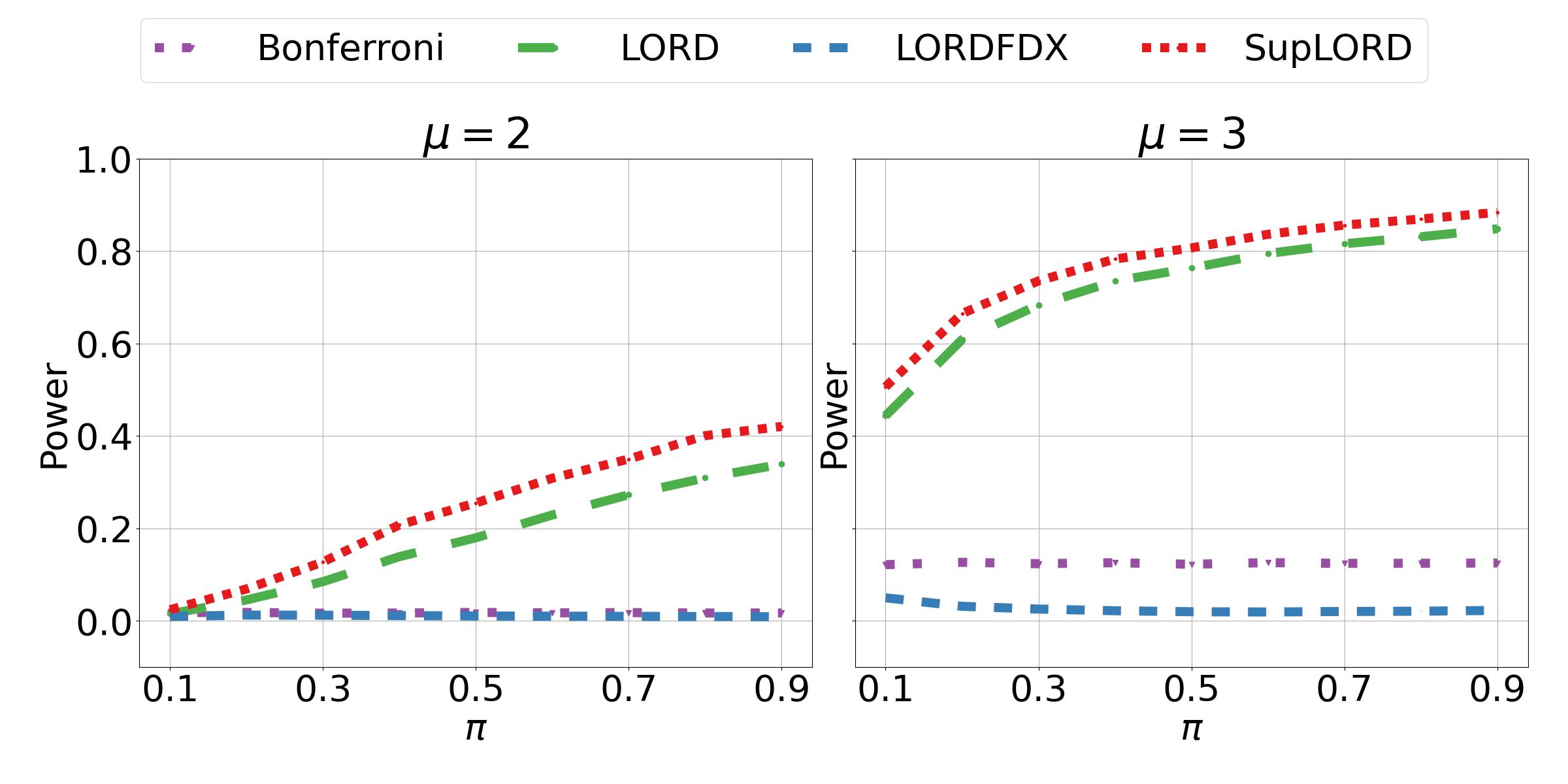} 
	\caption{Plot of non-null likelihood, \(\pi\), vs.  power for signal strengths of \(\mu \in \left\{2, 3\right\}\). \(\FDR\) is controlled at \(\ell =0.05\), and \(\SupLORD\) and \(\LORD\FDX\) are controlled at \(\FDX\) with \(\epsilon^*=0.15\)  at a level of \(\delta^*=0.05\). Experiment details are in \Cref{subsec:FDXExperiments}. \(\SupLORD\) consistently has higher power across all non-null likelihoods and signal strengths.}
	\label{fig:BaselineGaussian}
\end{figure}

\section{Dynamic Scheduling With Non-Monotone Spending Rules}
\label{sec:DynamicScheduling}

\paragraph{Under-utilized wealth leads to unnecessary conservativeness.}
Recall that the wealth value at each step, \(W(k)\), quantifies the error budget that can be allocated to the next alpha value. \Cref{subfig:IncreasingWealth} shows that, empirically, the wealth increases as more hypotheses are tested when using \steady\ with \(\SupLORD\) and \(\LORD\). To the best of our knowledge, this problem of \steady\ has not been addressed in previous work. The algorithm is not fully utilizing its error budget --- the rate at which the algorithm is making rejections and accumulating wealth is much higher than the rate which wealth is spent. While \aggressive\ does have better utilization of its wealth, it also does not distribute its spending as evenly as \steady, since it spends a constant percentage of its current wealth at each step. To remedy this, we derive a new spending schedule that uses wealth much more rapidly when wealth has been accumulated. At the same time, we do not want our algorithm to spend excessively when wealth is scarce. Thus, our new spending schedule is \textit{adaptive} to its wealth --- it adjusts its rate of spending to evenly spend as much wealth as possible while avoiding alpha-death.

\revision{At a technical level, the underutilization of wealth occurs because prior alpha-spending rules are ``monotone'' \eqref{eqn:MonotoneSchedule} functions of past rejections, effectively handicapping them by making them unable to adapt to the current wealth. This monotonicity restriction was primarily imposed in order to mathematically prove FDR control, instead of mFDR control, in a proof technique introduced by \citet{javanmard2018online}, and extended by all the followup works by Ramdas and coauthors. Our new proof technique avoids the requirement for monotonicity entirely, thus opening the possibility of designing smarter, wealth-adaptive spending rules, while still having FDX and FDR control.}

\subsection{Dynamic Scheduling Formulation}
To address this problem, we provide a method for altering \(\{\spendcoef_i\}_{i \in \naturals}\), called \textit{dynamic scheduling}, based on the wealth of the algorithm. The general paradigm of our new spending schedule is similar to the \steady\ spending schedule in \eqref{eqn:StaticSpendingSchedules}. In the \steady\ spending schedule, we can see that each alpha value is the sum of portions of the wealth boosts at each past rejection. The portion of the wealth boost of the \(i\)th rejection that the current alpha value spends is dependent on \(\gamma_{k - t_i}\), where \(k\) is our current time, and \(t_i\) is the time the \(i\)th rejection was made. Our new spending schedule follows the same paradigm. However, we use a spending sequence that changes based upon the wealth of the algorithm at the time of the rejection, which we denote \(\{\accelgamma_i\}_{i \in \naturals}\). To create this, we must first have a base spending sequence \(\{\spendcoef_i\}_{i \in \naturals}\). Then, we formulate this new spending sequence as follows:
\begin{align}
	\accelgamma_i(W, W(0)) &\equiv \begin{cases}
		\frac{\gamma_i^{\decaycoef \cdot \frac{W}{W(0)} \vee 1}}{\sum\limits_{j = 1}^\decaylen \gamma_j^{\decaycoef \cdot \frac{W}{W(0)} \vee 1}} & \text{if }\decaycoef \cdot \frac{W}{W(0)} > 1,\ i \leq \decaylen\ \text{(active)}\\
		0 & \text{if }\decaycoef \cdot \frac{W}{W(0)} > 1,\ i > \decaylen\ \text{(active and finished)}\\
		\gamma_i&\text{else (inactive)},
	\end{cases}
	\label{eqn:DynamicGammaFormulation}
\end{align}
where \(W\) and \(W(0)\) are the wealth of the algorithm at the time of rejection and the beginning, respectively. \(\accelgamma_i\) adapts to wealth by switching between the active and inactive cases. If the wealth is high when making a rejection, \(\accelgamma_i\) goes into the active case where it spends wealth at a more rapid pace. Otherwise, it falls back to spending according to \(\gamma_i\). \Cref{subfig:AccelCoefVisualization} illustrates this adaptivity. \(\decaycoef \in \reals^+\) controls the pace at which \(\accelgamma_i\) spends wealth as a function the algorithm's current wealth. \(\decaylen \in \naturals\) controls the length of time before \(\accelgamma_i\) has exhausted the wealth earned from its rejection --- the first \(\decaylen\) elements of \(\gamma_i\) are normalized to create \(\accelgamma_i\). Larger \(\decaycoef\) and smaller \(\decaylen\) increases spending in the active case. 
Note that \(\accelgamma_i\) satisfy conditions~\eqref{eqn:SpendingConditions}, making it a valid spending sequence.

\begin{figure}
    \centering
    \subfigure[Wealth of different GAI algorithms (setting specified in \Cref{subsec:FDXExperiments}). The line is the mean over 200 trials and the shaded area is a 95\% confidence set. \(\SupLORD\) and \(\LORD\) both have increasing wealth.]{
        \centering
        \includegraphics[width=0.45\textwidth]{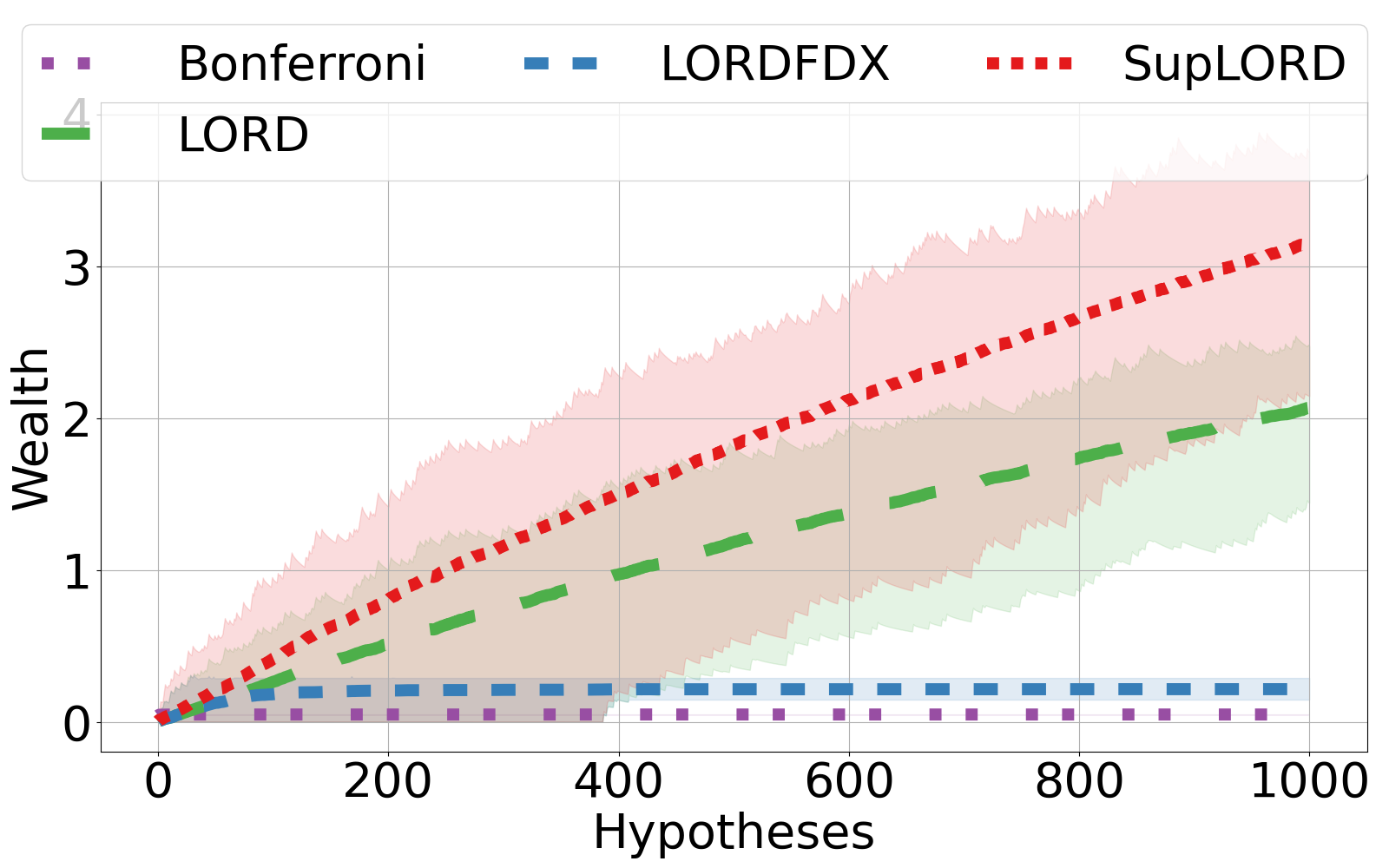}
        \label{subfig:IncreasingWealth}
    }\qquad\subfigure[Visualization of \(\accelgamma_i\) vs \(\spendcoef_i\) in \eqref{eqn:DynamicGammaFormulation}. \(\accelgamma_i\) is only active for \(\rho\) duration and is much larger than \(\gamma_i\) during that duration, which results in increased spending.]{
        \centering
        \includegraphics[width=0.45 \textwidth]{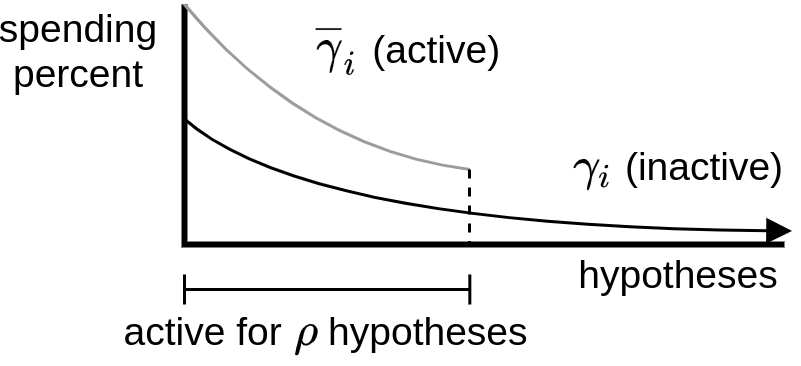}
        \label{subfig:AccelCoefVisualization}
    }
    \caption{Visualization of concepts pertaining to dynamic scheduling.}
    \label{fig:DynamicGammaPicture}
\end{figure}

Thus, using \(\accelgamma_i\) allows for the spending schedule to quickly use up entire boosts in wealth from rejections if the current wealth of the algorithm is large. We define our dynamic alpha schedule as:
\begin{align}
\alpha_k = \accelgamma_k(W(0), W(0))\cdot \beta_0 + \sum\limits_{i = 1}^{\left|\rejset_{k - 1}\right|} \accelgamma_{k - t_i}(W(t_i), W(0))\cdot \beta_{t_i}.
\end{align} We will show through simulation results that this method does solve the problem of wealth under-utilization and increases the power of \(\SupLORD\) as a result.

\subsection{Numerical Simulations and Real Data Experiments for Dynamic Algorithms}
\ifarxiv{}{\vspace{-5pt}}
\label{subsec:DynamicExperiments}
\paragraph{Spending schedules.}
We compare the differences in \steady, \aggressive, and dynamic scheduling on \(\SupLORD\). We choose \(\epsilon^*=0.15,\ \delta^*=0.05,\ r^*=30\) \revision{as the parameters for simulations, and \(r^*=100\) for the IMPC dataset} in all \(\SupLORD\) algorithms. We use the default boost sequence and spending sequence for all three schedules.
\ifarxiv{}{\vspace{-5pt}}
\paragraph{Results for comparing spending schedules.}
\begin{figure}
	\centering
	\subfigure[Non-null likelihood $\pi$ vs. power]{
		\includegraphics[width=0.58\textwidth]{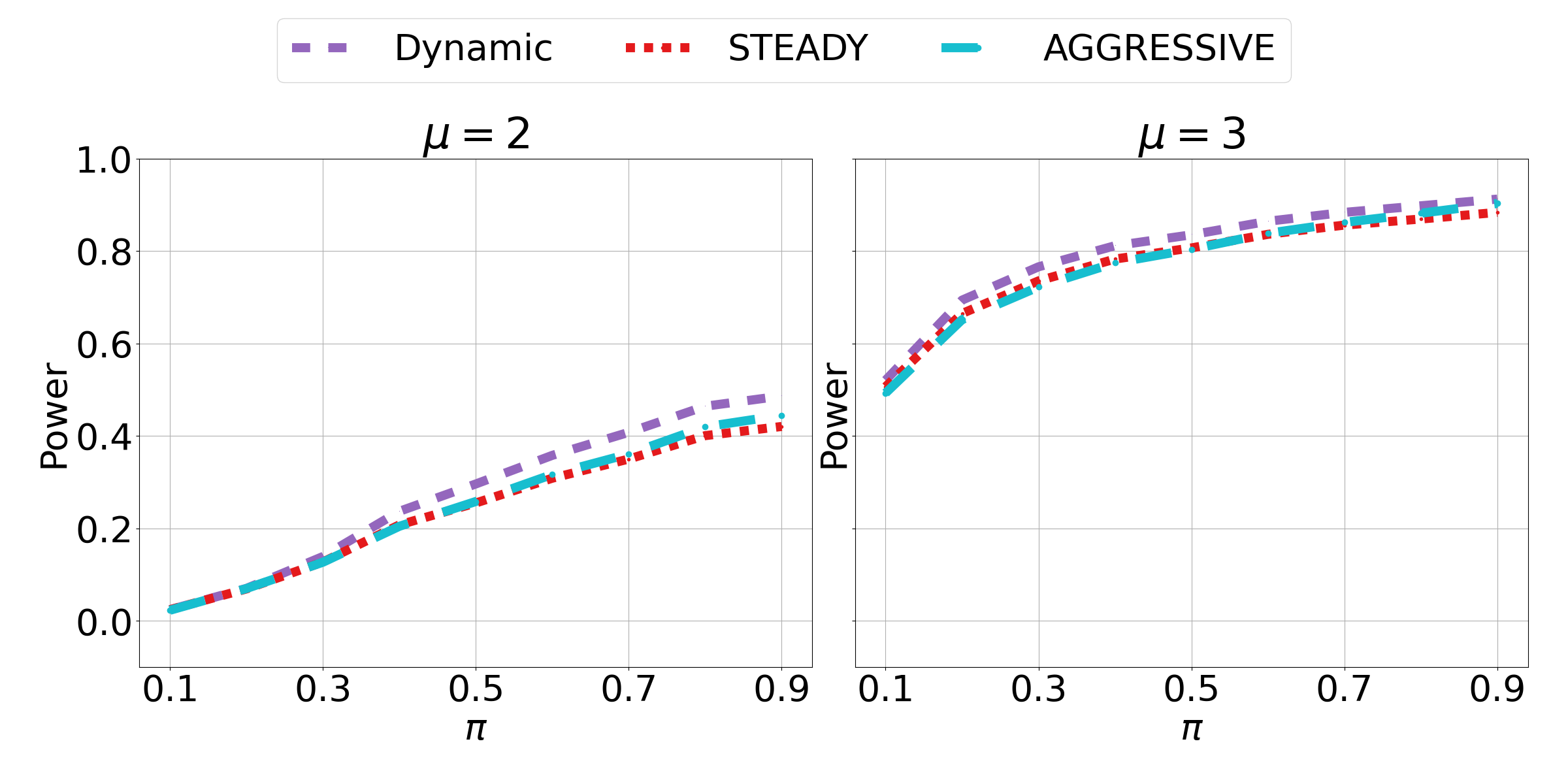}
		\label{subfig:DynamicGaussian}
	} \subfigure[Wealth vs. time]{
		\centering 
		\includegraphics[width=0.37\textwidth]{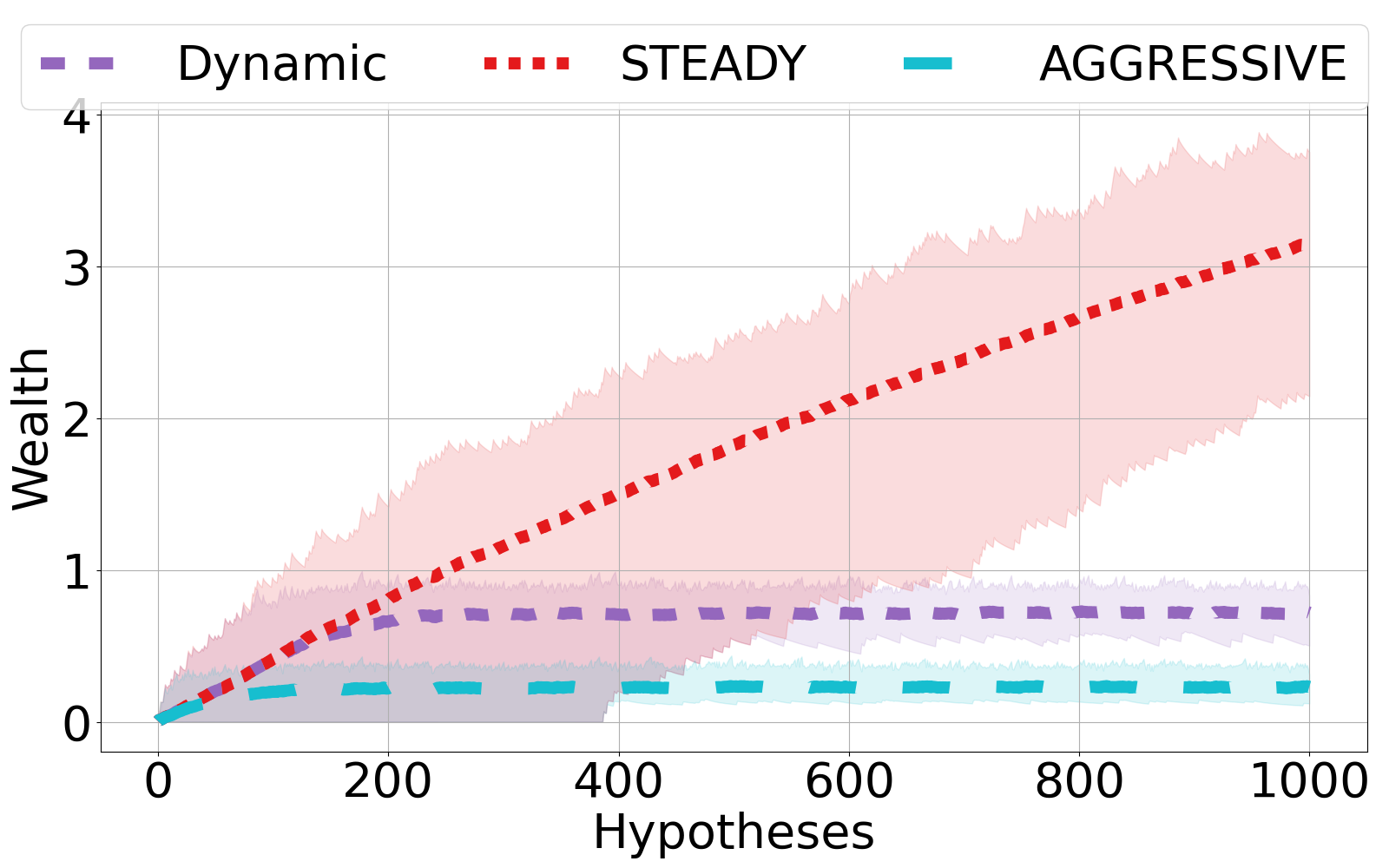}
		\label{subfig:DynamicGaussianWealth}
	}
	\caption{The left plot compares different spending schedule choices for \(\SupLORD\) plotted by likelihood of non-null hypotheses, \(\pi\), vs. power for signal strengths of \(\mu \in \left\{2, 3\right\}\) in the manner specified in \Cref{subsec:FDXExperiments}. Dynamic scheduling provides a consistent (albeit small) power improvement  over \steady\ and \aggressive\ \(\SupLORD\). The right plot shows that the dynamic scheduling does not hoard wealth like \steady, but saves more than \aggressive.}
\end{figure}

\begin{wraptable}{r}{0.3\linewidth}
    \centering
    \caption{\label{table:RealDataDynamic}\revision{Rejections made by each schedule on the IMPC data. Dynamic scheduling has the most rejections.}}
    \begin{tabular}{l|r}
         \revision{\textbf{Schedule}} & \revision{\textbf{Rejections}}\\
         \hline
         \revision{\steady} & \revision{18519} \\
         \revision{\aggressive} & \revision{15416} \\
         \revision{dynamic} & \revision{\textit{18793}}
    \end{tabular}
\end{wraptable}

\Cref{subfig:DynamicGaussian} shows that using a dynamic schedule does increase power when the signal and likelihood of the non-null hypotheses are larger, and does not do worse than the \steady\ or \aggressive\ when the signal and likelihood of non-nulls are low. This corresponds to the formulation of the \(\accelgamma_i\), as it falls back to \(\gamma_i\) if the algorithm's wealth is low. In the largest settings of non-null likelihood and signal, \aggressive\ catches up to dynamic scheduling in power. \revision{Admittedly, dynamic scheduling is limited in its power gains, since it only allocates existing wealth more effectively. Yet, it still outperforms both baseline schedules across simulation settings --- we illustrate that dynamic scheduling generally produces larger alpha values than either baseline schedule in \Cref{sec:DynamicSchedulingExp}. Further, a dynamic schedule also makes more rejections than other schedules on the IMPC dataset in \Cref{table:RealDataDynamic}, which is consistent with its power improvement in simulations. 
Thus, a dynamic schedule fully utilizes wealth and is a good default for practitioners.}

\ifarxiv{}{\vspace{-10pt}}
\section{Improving \(\FDR\) Control Through \(\supFD\) Control}
\label{sec:ControllingError}

Finally, we illustrate the connections between \(\FDR\), \(\FDX\), and \(\supFD\). We use additional results from \(\KR\) to prove that \(\SupLORD\) controls the \(\supFD\), and by extension, \(\FDR\). These \(\FDR\) results close the gap between \(\FDR\) control and \(\mFDR\) control~\eqref{def:mfdr}, which was first introduced by \FS\ as a proxy for \(\FDR\).  \FS\ showed that  \(\mFDR(\rejset_\tau)\), where \(\tau\) is a stopping time with respect to $\{\filtration_k\}$, could be controlled under assumption~\eqref{eqn:SuperUniform}. One example of a stopping time is the time of the $k$-th rejection, for some fixed $k \in \naturals$.
Successive works (\JM, \citealt{ramdas2017online}) have managed to show \(\FDR\) control with more restrictive assumptions of \eqref{eqn:IndSuperUniform} and monotone spending schedules.\footnote{Technically, \JM\ do have a theorem that proves \(\FDR\) control for dependent p-variables, but it is dependent on a strict uniformity assumption and also suffers from alpha-death since the proportion of wealth it can spend at each step is inversely proportional to the logarithm of the number of hypotheses that have been tested.} In addition, they also have only controlled at \(\FDR\) fixed times, unlike the stopping time control of \(\mFDR\) proved by \FS. Thus, \(\SupLORD\) completes the original goal of \FS\ by showing control, under the same assumptions, on the actual desired metric, \(\FDR\), for stopping times. To prove this, we first note \(\SupLORD\) provides a family of \(\FDX\) bounds by \Cref{fact:PosthocBound}.
\ifarxiv{}{\vspace{-5pt}}
\begin{corollary}[Corollary of \Cref{thm:SupLORDFDX}]
	Assuming conditional superuniformity \eqref{eqn:SuperUniform}, any GAI algorithm that satisfies conditions~\eqref{eqn:SupLORDBoostConditions} also ensures the following family of bounds for any \(\delta \in (0, 1)\):
	\ifarxiv{}{\vspace{-5pt}}
	\begin{align*}
		\FDX_{t_{r^*}}\left(\frac{\coef{\frac{1}{\delta}}}{\coef{\frac{1}{\delta^*}}} \epsilon^*\right) ~\leq~ \delta. 
	\end{align*} 
	\label{corollary:SupLORDFamilyFDXBound}
\end{corollary}
\ifarxiv{}{\vspace{-15pt}}
Before we state our control of \(\supFD\), we will introduce the following constant:
\ifarxiv{}{\vspace{-5pt}}
\begin{align*}
	c &= \int\limits_{0}^1 \coef{\frac{1}{\delta}}\ d \delta \approx 1.419 ~.
\end{align*}
\ifarxiv{}{\vspace{-15pt}}
\begin{theorem}
	Assuming conditional superuniformity \eqref{eqn:SuperUniform}, any GAI algorithm that satisfies conditions~\eqref{eqn:SupLORDBoostConditions} ensures the following bound holds:
	\ifarxiv{}{\vspace{-10pt}}
	\begin{align*}
		\supFD_{t_{r^*}} ~\leq~ \frac{c\epsilon^*}{ \coef{\frac{1}{\delta^*}}}.
	\end{align*}
	\label{thm:SupremumBound}
\end{theorem}
\ifarxiv{}{\vspace{-15pt}}
A detailed proof is provided in \Cref{sec:SupremumBoundProof} of the supplement. Intuitively, we can achieve this expectation by integrating over the family of \(\FDX\) bounds described in Corollary~\ref{corollary:SupLORDFamilyFDXBound}. Thus, this theorem establishes a connection between controlling \(\FDX\) and \(\supFD\) for \(\SupLORD\). Unlike prior \(\FDR\) bounds which are only valid at fixed times, \Cref{thm:SupremumBound} provides a bound on \(\FDR\) for stopping times. Formally, a stopping time \(\tau\) is a random variable where the event \(\tau = k\) is measurable in \(\filtration_k\) for all \(k\). Due to the supremum in \Cref{thm:SupremumBound}, we may derive the following corollary:
\ifarxiv{}{\vspace{-5pt}}
\begin{corollary}
	Assuming conditional superuniformity \eqref{eqn:SuperUniform}, \(\tau\) is a stopping time that satisfies \(\tau \geq  t_{r^*}\), any GAI algorithm that satisfies conditions~\eqref{eqn:SupLORDBoostConditions} ensures the following bounds hold:
	\begin{align*}
		\FDR(\rejset_\tau) ~\leq~ \frac{c\epsilon^*}{ \coef{\frac{1}{\delta^*}}}.
	\end{align*}
	\label{corollary:FDRStoppingTime}
\end{corollary} 
\ifarxiv{}{\vspace{-15pt}}
Thus, using \(\SupLORD\) provides the benefit of maintaining \(\FDR\) control even when the user may desire to adaptively stop testing hypotheses, and lends flexibility to its usage as a result. Since this result follows directly from \Cref{thm:SupremumBound}, it does not require any extra assumptions besides conditional superuniformity \eqref{eqn:SuperUniform}. In contrast, prior online FDR methods (\citealt{ramdas2017online,ramdas2018saffron}, \JM) require stronger assumptions of independence \eqref{eqn:IndSuperUniform} and monotone spending schedules \eqref{eqn:MonotoneSchedule}. Further, they only control \(\FDR\) at fixed times. Consequently, Corollary~\ref{corollary:FDRStoppingTime} improves upon prior work in both the generality of its assumptions and the strength of its guarantee. Additional \(\SupLORD\) guarantees are discussed in \Cref{sec:FDRSupLORDProof} (tighter \(\FDR\) control at fixed times) and \Cref{sec:mFDRSupLORDProof} (\(\mFDR\) control at stopping times).
\ifarxiv{}{\vspace{-0.1in}}
\section{Related Work on Online Multiple Testing}
\revision{Online multiple testing was introduced by \FS~\citep{foster2008alpha}, which formulated the problem and setup, and the first algorithm better than a Bonferroni bound through the paradigm of ``alpha-investing''. Notably, \FS\ introduced the notion of \(\mFDR\) and proved control at stopping times for \(\mFDR\) in this work because they found \(\FDR\) control difficult to prove. \cite{aharoni2014generalized} extended the guarantees from \(\FS\) to a broader class of algorithms. The \(\LORD\) algorithm formulated by \(\JM\)~\citep{javanmard2018online} is a particularly powerful type of GAI algorithm, and \(\JM\) showed it could control \(\FDR\) at all fixed times --- this was consequently the first substantial \(\FDR\) control result in the area. In fact, \cite{chen_power_2020-1} proved that \(\LORD\) is asymptotically as powerful as the best offline algorithm for controlling \(\FDR\). \cite{ramdas2017online} then relaxed the conditions for \(\FDR\) control in \(\LORD\) and formulated a family of algorithms, known as GAI++, that improved power over any GAI algorithm, including \(\LORD\).} For instances of application, \cite{xu2015infrastructure} provide a comprehensive overview of how online hypothesis testing plays a crucial in modern A/B testing. \cite{robertson2018online,robertson2019onlinefdr} apply the problem in a biomedical setting.

Several works have improved the power of the \(\LORD++\) algorithm \citep{tian2019addis, ramdas2018saffron}. This paper's methods are not specific to \(\LORD\) --- we may generalize dynamic scheduling to any generalized alpha-investing algorithm and the \(\FDX\) control obtained by \(\SupLORD\) can also be applied to algorithms with other \(\FDP\) estimators, like the one in \cite{ramdas2018saffron}. 

Other related work on online selective inference includes online methods for controlling family-wise error rate \citep{tian_online_2021}, false coverage rate \citep{weinstein2019online}, and \(\FDR\) in minibatches i.e.\ neighboring batches of hypotheses also have \(\FDR\) control \citep{zrnic2020power}. Controlling any of these metrics guarantees control of the \(\FDR\) as well. \cite{zrnic_asynchronous_2021} considers \(\FDR\) control under local dependence between p-values. \cite{gang2020structureadaptive} take a Bayesian approach to the problem and increase the power of their algorithms by exploiting local structure. 

\revision{All aforementioned works are primarily concerned with error control in expectation and thus are focused on controlling the \(\FDR\). In contrast, this paper introduces an first empirically powerful algorithm, \(\SupLORD\), for controlling the \(\FDP\) at a set level with high probability i.e.\ \(\FDX\). Further, control at stopping times has only been shown for \(\mFDR\) by \FS\ and \cite{aharoni2014generalized}. \(\SupLORD\) is the first algorithm to demonstrate any control of \(\FDR\) at stopping times. Thus, \(\SupLORD\) broadens the online multiple testing literature to account for \(\FDX\) control and fills the existing gap of controlling \(\FDR\) at stopping times.}
\section{Summary}
We have shown a new algorithm, \(\SupLORD\), that controls \(\FDX\) and is empirically more powerful by a significant amount than previous \(\FDX\) controlling algorithms. We also show that \(\SupLORD\) may simultaneously control \(\FDX\), \(\FDR\), and \(\supFD\) together. In addition, we make the observation that existing algorithms are not tight to their level of error control, and formulate dynamic scheduling, a framework for adapting to the wealth of an algorithm. We empirically show that dynamic scheduling increases power and reduces this wealth gap. Thus, we formulated the most powerful algorithms for \(\FDX\) control and brought insight to the interplay between multiple error metrics in this paper. \revision{Finally, our paper showed the first non-trivial algorithm with \(\FDR\) control at stopping times, filling a gap in the literature that had been noticed by \FS\ when they first proposed the problem.} We plan on implementing our methods and incorporating them into the \texttt{onlineFDR} R package \citep{robertson2019onlinefdr}, which contains current state-of-the-art online multiple testing algorithms. One line for future work is understanding how \(\decaycoef\) and \(\decaylen\) can be set optimally in dynamic scheduling, and whether there is choice of parameters that is provably most powerful or optimal in some sense.

\paragraph{Acknowledgements} We would like to thank our anonymous reviewers for their suggestions and questions. AR acknowledges support from NSF CAREER 1945266.


\bibliography{hypothesis}
\clearpage
\appendix
\section{\(\SupLORD\) General Formulation}
\label{sec:GeneralSupLORD}

We first define a more general form of \(\SupLORD\) that is valid for any \(\gamma^*, r^*\). We add an additional user parameter of \(a > 0\). We now redefine our estimators of \(\FDP\) to include \(a\).
\ifarxiv{}{\vspace{-0.1in}}
\begin{align*}
	\widehat{\FDP}_a(\rejset_k) &~\equiv~ \frac{\widehat{V}_k + a}{\left|\rejset_k\right|},\\
	\coefext{a}{\frac{1}{\delta}} &~\equiv~ \frac{\log(\frac{1}{\delta})}{a \log\left(1 + \frac{\log(\frac{1}{\delta})}{a}\right)} ~ ,\\
	\overline{\FDP}_a(\rejset_k) &~\equiv~ \coefext{a}{\frac{1}{\delta}} \cdot  \widehat{\FDP}_a(\rejset_k).
\end{align*}\ifarxiv{}{\vspace{-0.1in}}

We drop \(a\) from the notation when \(a = 1\). We first cite the more general form of \Cref{fact:PosthocBound}.

\begin{fact}[Detailed form of Theorem 4 from \(\KR\)]
	Let \(\{\alpha_k\}_{k \in \naturals}\) be any sequence of alpha values predictable with respect to filtration \(\filtration_{k}\) in definition~\eqref{eqn:Filtration}.
	Assuming conditional superuniformity \eqref{eqn:SuperUniform}, the following uniform bound holds:
	\begin{align*}
		\prob{\FDP(\rejset_k) \leq \overline{\FDP}_a(\rejset_k) \text{ for all }k \geq 1} \geq 1 - \delta.
	\end{align*}
	\label{fact:GeneralPosthocBound}
\end{fact}
\ifarxiv{}{\vspace{-0.3in}}

Consequently, our conditions on the boost sequence \(\{\beta_i\}_{i \in \naturals \cup \{0\}}\) is modified to account for \(a\):

\begin{subequations}
\begin{align}
    \sum\limits_{i = 0}^{t_{r^* - 1}} \beta_iR_i &\leq \frac{\epsilon^*r^*}{\coefext{a}{\frac{1}{\delta^*}}} - a \label{eqn:GeneralSumSupLordBoostCondition}\\
    \beta_i & \leq \frac{\epsilon^*}{\coefext{a}{\frac{1}{\delta^*}}}\ \text{for all }i > t_{r^* - 1}. \label{eqn:GeneralIncreaseSupLordBoostCondition}
\end{align}
\label{eqn:GeneralSupLORDBoostConditions}
\end{subequations}

The default \(\{\beta_i\}_{i \in \naturals\cup\{0\}}\) can be specified as follows:
\begin{align}
	\beta_i = \begin{cases}
		\frac{1}{r^*}\left(\frac{\epsilon^* r^*}{\coefext{a}{\frac{1}{\delta^*}}} - a\right) & i < r^* , \\
		\frac{\epsilon^*}{\coefext{a}{\frac{1}{\delta^*}}} & i \geq r^*
	\end{cases} \quad .
	\label{eqn:GeneralSimpleSupLORDB}
\end{align} Choosing \(a = 1\) recovers conditions~\eqref{eqn:SupLORDBoostConditions} and the default sequence in \eqref{eqn:DefaultBoostSupLORD}.

Naturally, the we then prove the following bound holds for the general form:
\begin{theorem}
	Given that conditional superuniformity \eqref{eqn:SuperUniform} holds, let \(0 < \epsilon^*, \delta^* < 1\), \(r^* \geq 1\) be a positive integer, and \(t_{r^*}\) be a random variable that is the time when the \(r^*\)th rejection is made. 
    A GAI algorithm with a boost sequence \(\{\beta_i\}_{i \in \naturals \cup \{0\}}\) that satisfies conditions~\eqref{eqn:GeneralSupLORDBoostConditions} guarantees:
    \begin{align*}
        \FDX_{t_r^*}^{\epsilon^*} \leq \delta^*.
    \end{align*}
    \label{thm:GeneralSupLORDFDX}
\end{theorem}

Similar to \Cref{sec:SupLORD}, \Cref{thm:GeneralSupLORDFDX} follows from \Cref{fact:GeneralPosthocBound} and a variant of Lemma~\ref{lemma:FDPBarLemma} for the general form \(\SupLORD\) --- Lemma~\ref{lemma:GeneralFDPBarLemma}. We prove Lemma~\ref{lemma:GeneralFDPBarLemma} in the sequel.

\subsection{Proof of Lemma~\ref{lemma:FDPBarLemma}}
\label{subsec:BoundedFDPProof}

First, we introduce this general form version of Lemma~\ref{lemma:FDPBarLemma}:
\begin{lemma}
    An GAI algorithm with a boost sequence \(\{\beta_i\}_{i \in \naturals \cup \{0\}}\) that satisfies conditions~\eqref{eqn:GeneralSupLORDBoostConditions} will ensure the following:
    \begin{align*}
        \overline{\FDP}(\rejset_k) \leq \epsilon^*
    \end{align*} for all \(k \geq t_{r^*}\) such that the \(k\)th hypothesis is rejected.
    \label{lemma:GeneralFDPBarLemma}
\end{lemma} 
\begin{proof}
Let \(k = t_r\) where \(r \geq r^*\). \(|\rejset_{t_r - 1}| = |\rejset_{t_r}| - 1\) since a new rejection is made at \(t_r\):
\begin{align*}
	\overline{\FDP}_a(\rejset_{t_r}) &= \coefext{a}{\frac{1}{\delta^*}} \cdot \left(\frac{\sum\limits_{i = 1}^{t_r}\alpha_i + a}{\left|\rejset_{t_r}\right|}\right) \tag*{by definition}\\
    &= \coefext{a}{\frac{1}{\delta^*}} \cdot \left(\frac{\beta_0 + \sum\limits_{i = 1}^{t_r - 1} \beta_i - W(k - 1) + \alpha_{t_r} + a}{\left|\rejset_{t_r}\right|}\right) \tag*{by \eqref{eqn:AlphaInvestingWealth}}\\
    &\leq \coefext{a}{\frac{1}{\delta^*}} \cdot \left(\frac{\beta_0 + \sum\limits_{i = 1}^{t_r - 1} \beta_i + a}{\left|\rejset_{t_r}\right|}\right) \tag*{by \eqref{eqn:WealthInvariant}}\\
    &= \coefext{a}{\frac{1}{\delta^*}} \cdot \left(\frac{\frac{\epsilon^*r^*}{\coefext{a}{\frac{1}{\delta^*}}} + \frac{\epsilon^*(\left|\rejset_{t_r - 1}\right| - (r^* - 1))}{\coefext{a}{\frac{1}{\delta^*}}}}{\left|\rejset_{t_r}\right|}\right) \tag*{by conditions~\eqref{eqn:GeneralSupLORDBoostConditions}}\\
    &\leq\coefext{a}{\frac{1}{\delta^*}} \cdot \left(\frac{\frac{\epsilon^*r^*}{\coefext{a}{\frac{1}{\delta^*}}} + \frac{\epsilon^*(\left|\rejset_{t_r}\right| - 1 - (r^* - 1))}{\coefext{a}{\frac{1}{\delta^*}}}}{\left|\rejset_{t_r}\right|}\right) \tag*{since a rejection is made at time \(t_r\)}\\
    &= \epsilon^*.
\end{align*}
Thus, we have proven Lemma~\ref{lemma:GeneralFDPBarLemma}. 
\end{proof}

Lemma~\ref{lemma:FDPBarLemma} follows directly from setting \(a = 1\).

\subsection{Tradeoffs in the choice of $a$}
\label{subsec:Tradeoffs}

Let \(w_0\) be the upper bound in condition~\eqref{eqn:GeneralSumSupLordBoostCondition} and \(b\) be the upper bound in condition~\eqref{eqn:GeneralIncreaseSupLordBoostCondition}:
\begin{align*}
    w_0 ~&\equiv~ \frac{\epsilon^*r^*}{\coefext{a}{\frac{1}{\delta^*}}} - a,\\
    b ~&\equiv~ \frac{\epsilon^*}{\coefext{a}{\frac{1}{\delta^*}}}\ \text{for all }i > t_{r^* - 1}.
\end{align*}

The offset, \(a\), in \(\overline{\FDP}_a\) plays a crucial role in determining the tradeoff between \(w_0\) and the bound on earning per rejection from \(r^*\) onwards, \(b\). Given the user chosen parameters \(\epsilon^*\), \(\delta^*\), and \(r^*\), we provide a principled way of choosing \(a\) such that both quantities can be maximized.

First, note that the upper bound in condition~\eqref{eqn:GeneralIncreaseSupLordBoostCondition}, \(b\), is an increasing function of \(a\) that asymptotically approaches \(\epsilon^*\). Thus, we cannot maximize it. However, we may maximize \(w_0\), the upper bound in condition~\eqref{eqn:GeneralSumSupLordBoostCondition}, which is equivalent to solving the following optimization problem:
\begin{align*}
	\max_{a > 0} ~ ~ 
	\frac{\epsilon^*r^*}{\coefext{ a}{\frac{1}{\delta^*}}} - a ~ .
\end{align*}  Since the objective is a smooth concave function with respect to \(a\), by taking the derivative of the objective and solving for 0, we can determine the maximum \(w_0\) is achieved when the following condition holds:
\begin{align}
	\log\left(1 + \frac{\log\left(\frac{1}{\delta^*}\right)}{a}\right) - \frac{\log\left(\frac{1}{\delta^*}\right)}{a + \log\left(\frac{1}{\delta^*}\right)} = \frac{\log\left(\frac{1}{\delta^*}\right)}{\epsilon^* r^*}.
	\label{eqn:Canonicala}
\end{align}

\begin{figure}[!h]
	\centering
	\includegraphics[width=\columnwidth]{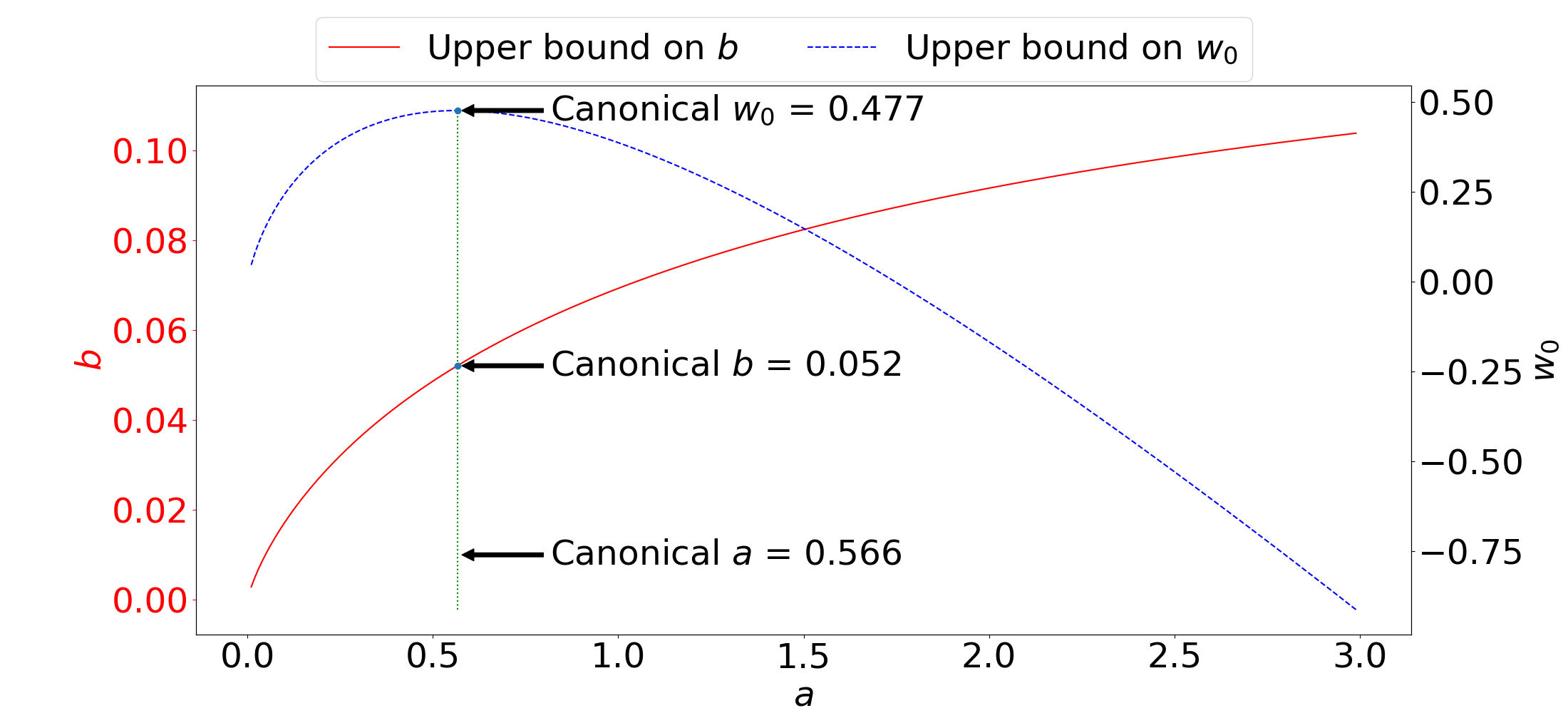}
	
	\caption{Tradeoff between \(w_0\) and \(b \) as a function of \(a\) given \(\delta^*=0.05, \epsilon^*=0.15, r^*=20\). Note that \(b\) is increasing with respect to \(a\), while \(w_0\) has a maximum. We also denote the canonical triple as a triple \((a, b, w_0)\), that can be selected to simultaneously maximize \(b\) and \(w_0\).}
	\label{fig:InitGainTradeoff}
\end{figure}

We illustrate in \Cref{fig:InitGainTradeoff} the relationship between \(a, b,\) and \(w_0\) for an example setting of \(\epsilon^*=0.15, \delta^*=0.05,\) and \(r^*=20\). We define the canonical triple for \((a, b, w_0)\) as the \(b\) and \(w_0\) corresponding to the \(a\) that is the positive solution to problem~\eqref{eqn:Canonicala}. Note that for all \(a\) smaller than the canonical \(a\), we may strictly improve any possible \(b\) and \(w_0\) under that \(a\) by selecting the canonical triple. On the other hand, the \(b\) and \(w_0\) for any \(a\) larger than the canonical \(a\) form an optimal Pareto-frontier for trading off \(b\) and \(w_0\). Thus, the canonical \(a\) can be treated as the default choice for \(a\) and consequently, the user may avoid tuning additional hyperparameters. In addition, we may also define a default \(\{\beta_i\}_{i \in \naturals \cup \{0\}}\) w.r.t.\ the canonical triple:
\begin{align}
    \beta_i = \begin{cases}
		\frac{1}{r^*}\left(\frac{\epsilon^* r^*}{\coefext{a}{\frac{1}{\delta^*}}} - a\right) & i \leq t_{r^* - 1}\\
		\frac{\epsilon^*}{\coefext{a}{\frac{1}{\delta^*}}} & i > t_{r^* - 1}.
	\end{cases}
\end{align} This default \(\{\beta_i\}_{i \in \naturals \cup \{0\}}\) is tight with conditions~\eqref{eqn:GeneralSupLORDBoostConditions} for the canonical \(a\), and can be considered as the default \(\{\beta_i\}_{i \in \naturals \cup \{0\}}\) sequence for the general \(\SupLORD\) form.

\section{Proof of \Cref{thm:SupremumBound}}
\label{sec:SupremumBoundProof}

First, we define the following term:
\begin{align*}
	c_a = \int\limits_0^1 \coefext{a}{\frac{1}{\delta}}\ d\delta.
\end{align*} Note that the \(c\) in \Cref{sec:ControllingError} is defined with $a=1$, in which case \(c = c_1 \leq 1.42\). 
\begin{theorem}
	Assuming conditional superuniformity \eqref{eqn:SuperUniform}, any GAI algorithm that satisfies conditions~\eqref{eqn:GeneralSupLORDBoostConditions}:
	\begin{align*}
		\supFD_{t_{r^*}} ~\leq~ \frac{c_a\epsilon^*}{ \coef{\frac{1}{\delta^*}}}.
	\end{align*}
	\label{thm:GeneralSupremumBound}
\end{theorem}
\Cref{thm:SupremumBound} can be seen as an instance of this theorem where \(a = 1\).

To prove this theorem, we must prove a preliminary fact.
\begin{fact}[Corollary 1 from an \href{https://arxiv.org/pdf/1803.06790v3.pdf}{earlier revision} of \(\KR\)]
	Let \(\{\alpha_k\}_{k \in \naturals} \) be any
	sequence of alpha values predictable with respect to filtration \(\filtration_k\) in definition~\eqref{eqn:Filtration}. Assume conditional superuniformity \eqref{eqn:SuperUniform} is true, and let \(a > 0\) be fixed. Then, we obtain the following expected supremum ratio bound:
	\begin{align*}
		\expect\left[
		\sup_{k \in \naturals}\ \frac{\FDP(\rejset_k)}{\widehat{\FDP}_a(\rejset_k)}\right] \leq c_a.
	\end{align*}
	\label{fact:ExpectedSupremumBound}
\end{fact}

We include a proof of this fact for completeness. First, we introduce another lemma from \(\KR\):

\begin{fact}[Lemma 1 from \(\KR\)]
	Let the conditions of \Cref{fact:ExpectedSupremumBound} be satisfied. Then, the following is true:
	\begin{align*}
		\prob{\sup_{k \in \naturals}\ \frac{\FDP(\rejset_k)}{\widehat{\FDP}_a(\rejset_k)} > x} ~\leq~ \exp\left(-a\theta_xx\right),
	\end{align*} where \(\theta_x\) is the unique solution to \(\exp(\theta) = 1 + \theta x\).
	\label{fact:MartingaleLemma}
\end{fact}

\begin{proof}[Proof of \Cref{fact:ExpectedSupremumBound}]
First, we characterize \(x\) when \(\exp(-a\theta_x x) = \delta\) in \Cref{fact:MartingaleLemma}:
\begin{align*}
	\theta_x x &= -\frac{\log(\delta)}{a},\\
	\theta_x &= \log\left(-\frac{\log(\delta)}{a} + 1\right). \tag*{by \(\exp(\theta_x) = 1+ \theta_x x\)}
\end{align*} We now substitute this equality for \(\theta_x\) into the following:
\begin{align*}
	\exp(\theta_x) &= 1 + \theta_x x,\\
	x&=\frac{\exp(\theta_x) - 1}{\theta_x}\\
	&= \frac{\frac{-\log(\delta)}{a}}{\log\left(-\frac{\log(\delta)}{a} + 1\right)}\\
	&=\frac{\log\left(\frac{1}{\delta}\right)}{a \log\left(1 + \frac{\log\left(\frac{1}{\delta}\right)}{a}\right)}\\
	&=\coefext{a}{\frac{1}{\delta}}.
\end{align*}

Now, we bound the expectation:
\begin{align*}
	\expect\left[\sup_{k \in \naturals}\ \frac{\FDP(\rejset_k)}{\widehat{\FDP}_a(\rejset_k)}\right] &= \int\limits_0^\infty \prob{\sup_{k \in naturals}\ \frac{\FDP(\rejset_k)}{\widehat{\FDP}_a(\rejset_k)} > x}\ dx \tag*{by quantile formula for expectation}\\
	&\leq \int\limits_0^\infty \exp(-a\theta_x x)\ dx \tag*{by \Cref{fact:MartingaleLemma}}\\
	&=\int\limits_{1}^0 \delta\ d \coefext{a}{\frac{1}{\delta}} \tag*{substitute in \(\delta\) and \(\exp(-a\theta_xx)\)}\\
	&=0 \cdot \coefext{a}{0} - \coefext{a}{1} + \int\limits_{0}^1 \coefext{a}{\frac{1}{\delta}}\ d\delta.
\end{align*}

To simplify the last expression, we observe the following two facts:
\begin{align*}
	\underset{\delta\rightarrow 0^+}{\lim} \delta \cdot \coefext{a}{\frac{1}{\delta}} ~&=~ 0,\\
	\underset{\delta\rightarrow 1^-}{\lim} \delta \cdot \coefext{a}{\frac{1}{\delta}} ~&=~ 1.
\end{align*} Thus, we can drop \(0 \cdot \coefext{a}{0}\), since it is 0, and \(- 1 \cdot \coefext{a}{1}\) since it negative, and we achieve our desired upper bound.
\end{proof}

We proceed to prove \Cref{thm:GeneralSupremumBound}.	
\begin{align*}
	\supFD_{t_{r^*}} &= \expect\left[\underset{k \geq t_{r^*}}{\sup}~\FDP(\rejset_k)\right]\\
	&=\expect\left[\underset{r \geq r^*}{\sup}~\FDP(\rejset_{t_r})\right]\\
	&\leq \expect\left[\underset{r \geq r^*}{\sup}~\frac{\FDP(\rejset_{t_r})}{\widehat{\FDP}_a(\rejset_{t_r})}\right] \cdot \left(\underset{r \geq r^*}{\sup}\ \widehat{\FDP}_a(\rejset_{t_r})\right) \tag*{relaxation of \(\sup\)}\\
	&\leq \frac{c_a}{\coefext{a}{\frac{1}{\delta^*}}} \cdot \left(\underset{r \geq r^*}{\sup}\ \overline{\FDP}_a(\rejset_k)\right) \tag*{by \Cref{fact:ExpectedSupremumBound} and definition of \(\widehat{\FDP}_a(\rejset_k)\)}\\
	&\leq \frac{c_a \epsilon^*}{\coefext{a}{\frac{1}{\delta^*}}}. \tag*{by Lemma~\ref{lemma:GeneralFDPBarLemma}}
\end{align*} Thus, we have shown that the general \(\SupLORD\) algorithm controls \(\supFD\) as claimed.

\section{Tighter Fixed Time \(\FDR\) Control}
\label{sec:FDRSupLORDProof}

A tighter, fixed time \(\FDR\) bound exists for \(\SupLORD\) with additional constraints:
\begin{theorem}
	 Assuming independent superuniformity \eqref{eqn:IndSuperUniform}, any GAI algorithm with a monotone spending schedule \eqref{eqn:MonotoneSchedule} that satisfies condition~\eqref{eqn:GeneralSupLORDBoostConditions} and
	 \begin{align}
	     \beta_0 + \beta_1 \leq \frac{\epsilon^*}{\coefext{a}{\frac{1}{\delta^*}}},
	     \label{eqn:FixedFDRBetaBound}
	 \end{align}
	can ensure the following fixed time \(\FDR\) bound:
	\ifarxiv{}{\vspace{-0.1in}}
	\begin{align*}
		\FDR_1 ~\leq~ \frac{\epsilon^*}{\coefext{a}{\frac{1}{\delta^*}}}.
	\end{align*}
	\label{thm:FDRSupLORD}
\end{theorem}
\ifarxiv{}{\vspace{-0.25in}} This result uses the same proof techniques as \cite{ramdas2017online, ramdas2018saffron, tian2019addis} for \(\FDR\) control, hence the much stronger assumptions and control only at fixed times. Thus, we can improve our \(\FDR\) control by a factor \(c_a\) of the cost of usability --- independence is not necessarily a practical assumption, and control only at fixed times denies the flexibility of early stopping to the user.
\begin{proof}
To prove \Cref{thm:FDRSupLORD}, we utilize the following theorem from \cite{ramdas2017online}:
\begin{fact}[Theorem 1 from \citealt{ramdas2017online}]
	Assuming independent superuniformity \eqref{eqn:IndSuperUniform}, any GAI algorithm with a monotone spending schedule \cref{eqn:MonotoneSchedule} with boost sequence \(\{\beta\}_{i \in \naturals \cup \{0\}}\) controls \(\FDR\) at any fixed time:
	\begin{align*}
		\FDR_1 \leq \max(\{\beta_0 + \beta_1\} \cup \{\beta_i\}_{i \in \naturals \cup \{0\}}).
	\end{align*}
	\label{fact:GAI++}
\end{fact}

Thus, \Cref{thm:FDRSupLORD} follows directly from condition~\eqref{eqn:FixedFDRBetaBound} and \Cref{fact:GAI++}.
\end{proof}
An example specification for \(\{\beta_i\}_{i \in \naturals \cup \{0\}}\) that satisfies the conditions of \Cref{thm:FDRSupLORD} is as follows:
\begin{align}
	\beta_i = \begin{cases}
		\frac{1}{2r^*}\left(\frac{\epsilon^* r^*}{\coefext{a}{\frac{1}{\delta^*}}} - a\right) & i = 0\\
		\frac{1}{2r^*}\left(\frac{\epsilon^* r^*}{\coefext{a}{\frac{1}{\delta^*}}} - a \cdot \ind{r^* > 1}\right) & i\leq t_1\\
		\frac{1}{r^*}\left(\frac{\epsilon^* r^*}{\coefext{1}{\frac{1}{\delta^*}}} - a\right) & t_1 < i \leq t_{r^* - 1}\\
		\frac{\epsilon^*}{\coefext{1}{\frac{1}{\delta^*}}} & i > t_{r^* - 1}
	\end{cases} ~ ~ .
	\label{eqn:DefaultmFDRBoostSequence}
\end{align} We simply reduce \(\beta_0 + \beta_1\) such that they sum to at most the upper bound in condition~\eqref{eqn:GeneralIncreaseSupLordBoostCondition}.

\section{\(\mFDR\) Control for \(\SupLORD\)}
\label{sec:mFDRSupLORDProof}

\begin{fact}[Theorem 2 from \citealt{ramdas2017online}]
Let \(\tau\) be a stopping time. Assuming conditional superuniformity \eqref{eqn:SuperUniform}, any GAI algorithm with boost sequence \(\{\beta_i\}_{i \in \naturals \cup \{0\}}\) ensures the following:
\begin{align*}
    \mFDR(\rejset_\tau) \leq \max(\{\beta_0 + \beta_1\} \cup \{\beta_i\}_{i \in \naturals \cup \{0\}}).
\end{align*}
\label{fact:mFDRControl}
\end{fact}
Essentially, \cite{ramdas2017online} prove that the same algorithm can control \(\FDR\) at fixed times under p-variable independence and monotone scheduling, and \(\mFDR\) at stopping times under conditional superuniformity. Thus, we derive \(\mFDR\) control for \(\SupLORD\):
\begin{theorem}
    Let \(\tau\) be a stopping time. Assuming conditional superuniformity \eqref{eqn:SuperUniform}, any GAI that satisfies condition~\eqref{eqn:GeneralSupLORDBoostConditions} and
	 \begin{align}
	     \beta_0 + \beta_1 \leq \frac{\epsilon^*}{\coefext{a}{\frac{1}{\delta^*}}},
	     \label{eqn:FixedFDRBetaBound}
	\end{align} can ensure the following fixed time \(\FDR\) bound:
	\ifarxiv{}{\vspace{-0.1in}}
	\begin{align*}
		\mFDR(\rejset_\tau) ~\leq~ \frac{\epsilon^*}{\coefext{a}{\frac{1}{\delta^*}}}.
	\end{align*}
	\label{thm:mFDRSupLORD}
\end{theorem}

\Cref{thm:mFDRSupLORD} follows directly from conditions~\eqref{eqn:GeneralSupLORDBoostConditions} and from \Cref{fact:mFDRControl}. We can use the boost sequence in \eqref{eqn:DefaultmFDRBoostSequence} as an example boost sequence that would give \(\SupLORD\) \(\mFDR\) control, since it also suffices for the assumptions in \Cref{thm:mFDRSupLORD}

\ifarxiv{\clearpage}{}
\section{Additional Numerical Simulations}
\label{sec:AdditionalExperiments}

\subsection{\(\FDR\) and \(\supFD\) in \(\FDX\) Control Simulations}
\label{subsec:ErrorFDXExperiment}

\Cref{fig:FDRErrorFDXExperiment} shows that despite the nonexistent power \(\LORD\FDX\) possessed relative to \(\SupLORD\), its error in \(\FDR\) is surprisingly close to \(\SupLORD\). On other hand, \(\LORD\FDX\) has very low \(\supFD\) in \Cref{fig:SupFDErrorFDXExperiment}.
 
\begin{figure}[h]
    \centering
    \includegraphics[width=0.7\textwidth]{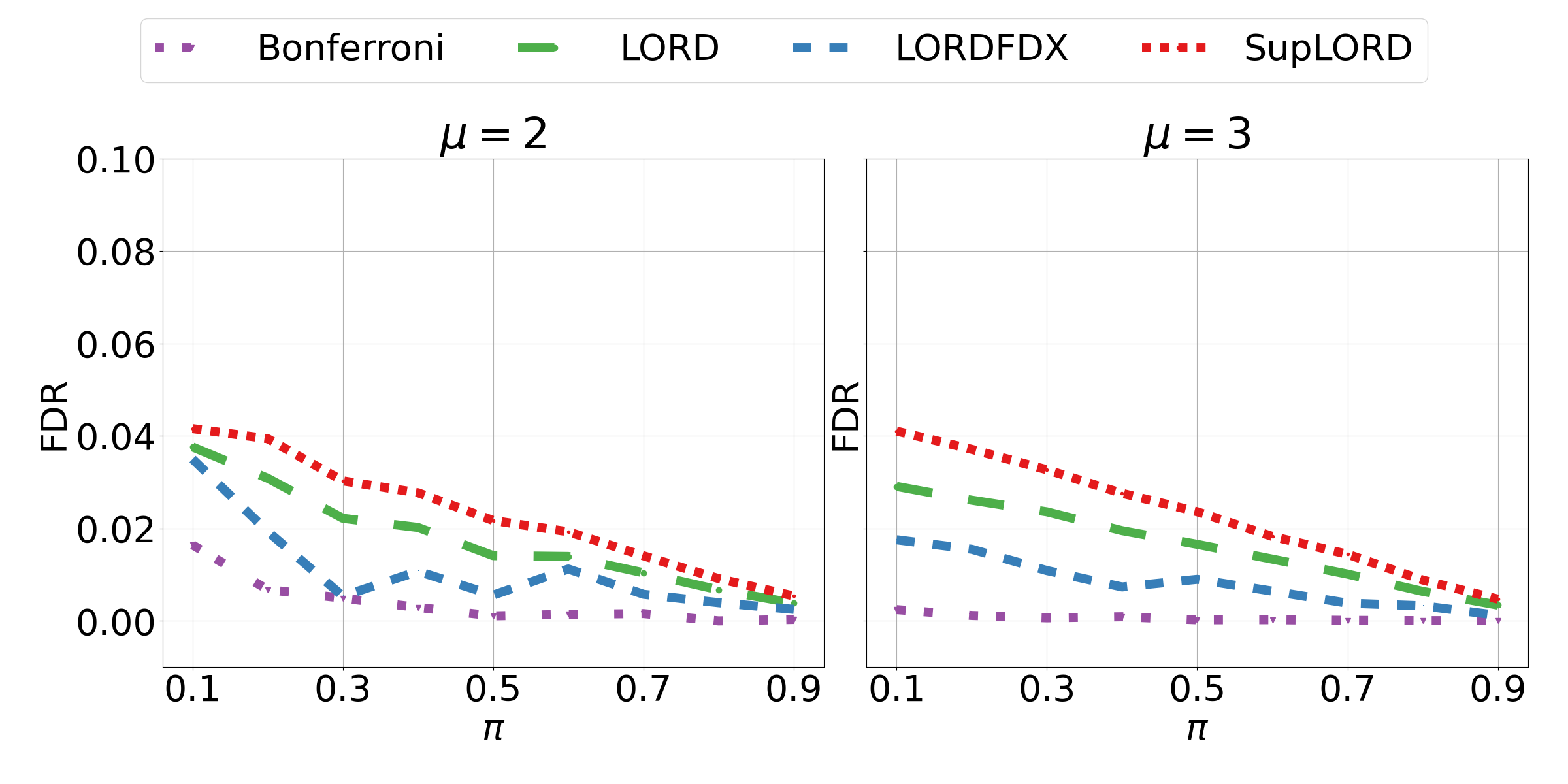}
    \caption{Plot of non-null likelihood, \(\pi\), vs.  \(\FDR\) for signal strengths of \(\mu \in \left\{2, 3\right\}\) in the constant data setting. \(\FDR\) is controlled at \(\ell =0.05\), and \(\SupLORD\) and \(\LORD\FDX\) are controlled at \(\FDX\) with \(\epsilon^*=0.15\)  at a level of \(\delta^*=0.05\). Experiment details in \Cref{subsec:FDXExperiments}. \(\SupLORD\) has \(\FDR\) that is somewhat appropriate to its results in \Cref{thm:FDRSupLORD}.}
    \label{fig:FDRErrorFDXExperiment}
\end{figure}

\begin{figure}[h]
    \centering
    \includegraphics[width=0.7\textwidth]{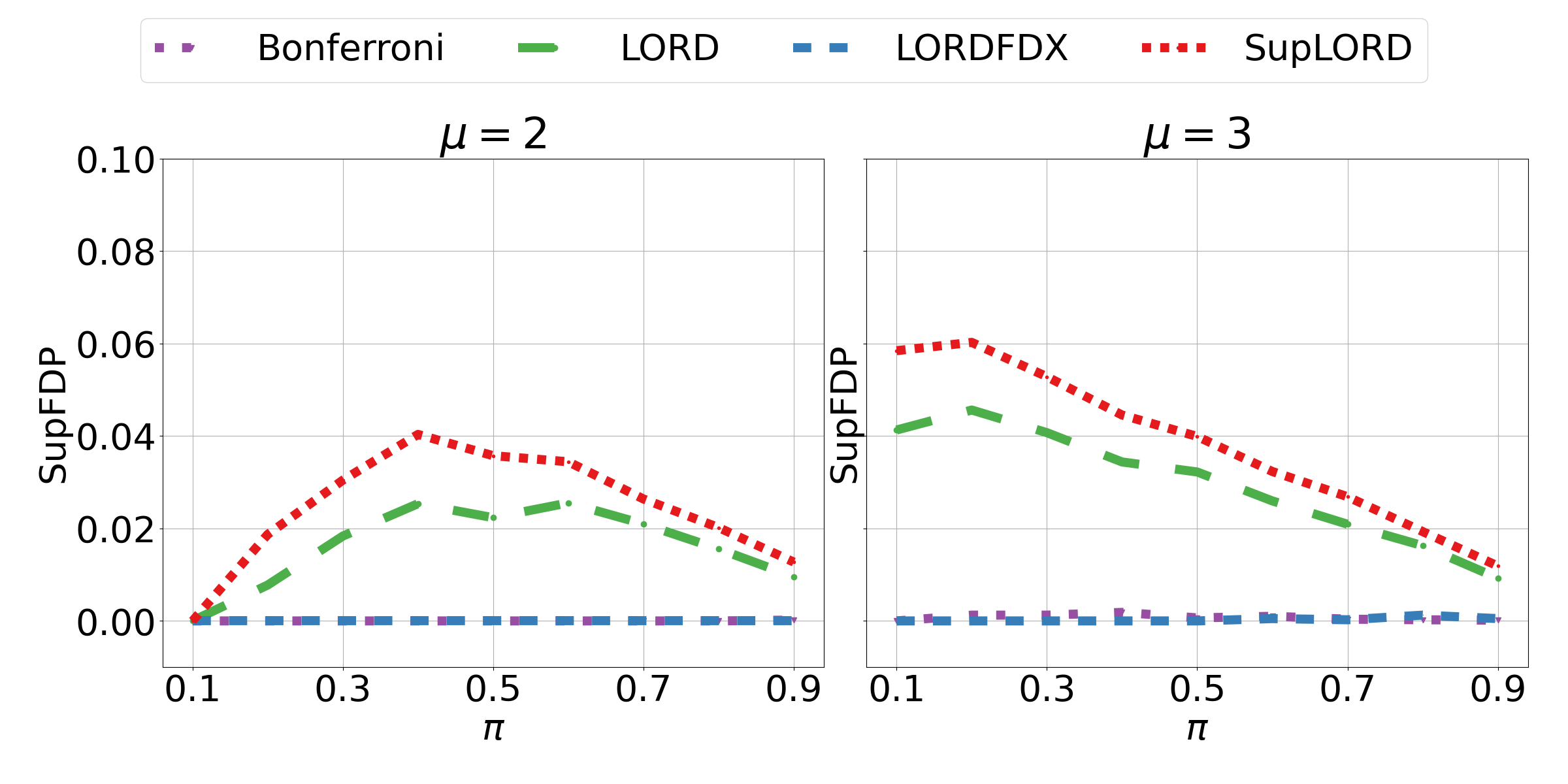}
    \caption{Plot of non-null likelihood, \(\pi\), vs.  \(\supFD_{30}\) for signal strengths of \(\mu \in \left\{2, 3\right\}\) in the constant data setting. \(\FDR\) is controlled at \(\ell =0.05\), and \(\SupLORD\) and \(\LORD\FDX\) are controlled at \(\FDX\) with \(\epsilon^*=0.15\)  at a level of \(\delta^*=0.05\). Experiment details in \Cref{subsec:FDXExperiments}. \(\SupLORD\) has \(\supFD_{30}\) below what is expected by \Cref{thm:SupremumBound}.}
    \label{fig:SupFDErrorFDXExperiment}
\end{figure}

\begin{figure}[!h]
    \centering
    \includegraphics[width=0.7\textwidth]{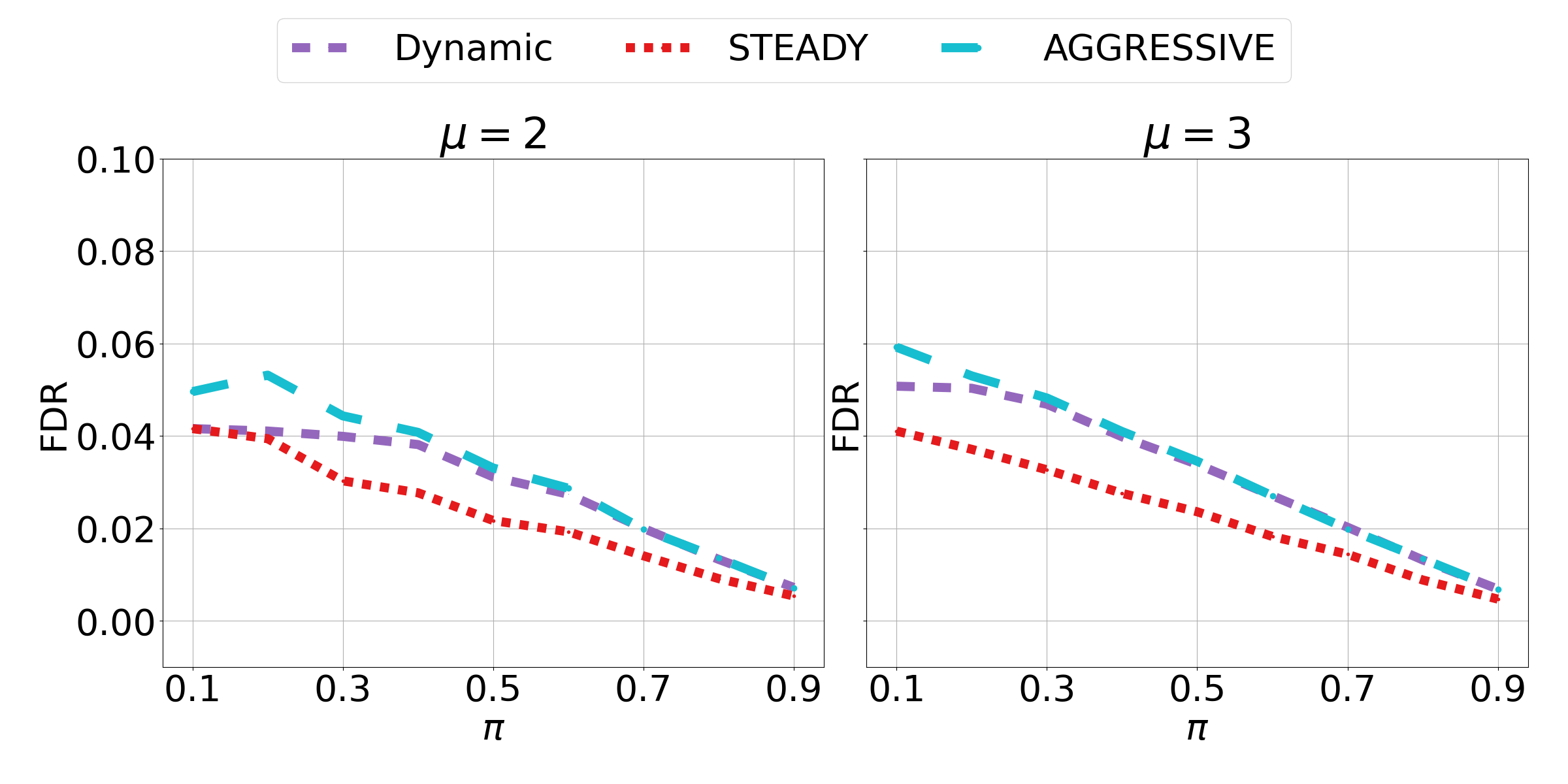}
    \caption{Plot of non-null likelihood, \(\pi\), vs.  \(\FDR\) for signal strengths of \(\mu \in \left\{2, 3\right\}\) in the constant data setting. \(\SupLORD\) is controlled at \(\FDX\) with \(\epsilon^*=0.15\)  at a level of \(\delta^*=0.05\). Data details in \Cref{subsec:FDXExperiments}, experiment details in \Cref{subsec:DynamicExperiments}. All schedules have similar \(\FDR\).}
    \label{fig:FDRErrorDynamicExperiment}
\end{figure}

\begin{figure}[!h]
    \centering
    \includegraphics[width=0.7\textwidth]{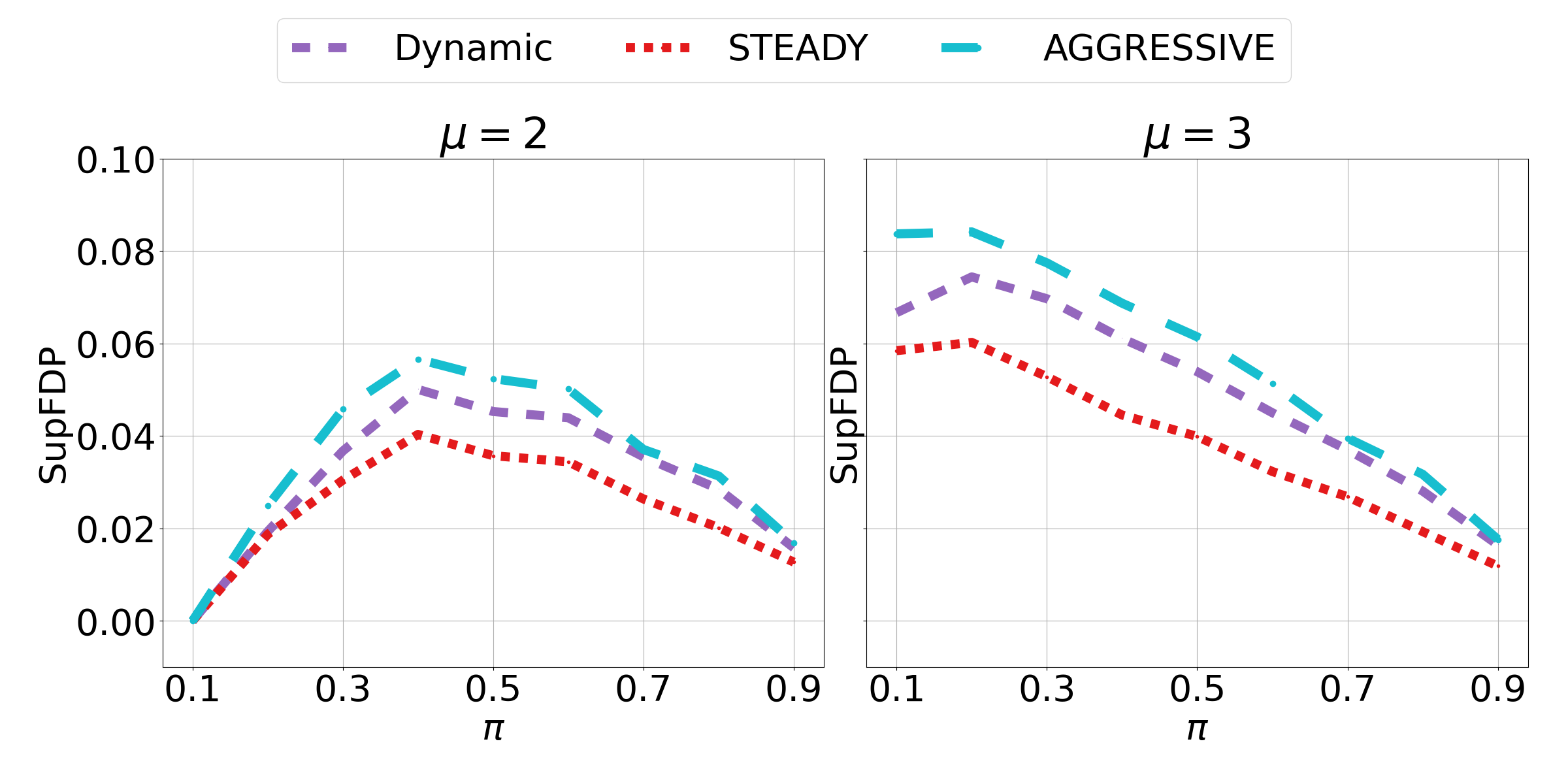}
    \caption{Plot of non-null likelihood, \(\pi\), vs.  \(\supFD_{30}\) for signal strengths of \(\mu \in \left\{2, 3\right\}\) in the constant data setting. \(\FDR\) is controlled at \(\ell =0.05\), and \(\SupLORD\) and \(\LORD\FDX\) are controlled at \(\FDX\) with \(\epsilon^*=0.15\)  at a level of \(\delta^*=0.05\). Data details in \Cref{subsec:FDXExperiments}, experiment details in \Cref{subsec:DynamicExperiments}. All schedules have similar \(\supFD\).}
    \label{fig:SupFDErrorDynamicExperiment}
\end{figure}
\clearpage

\subsection{Hidden Markov Model (HMM)}
In addition to the standard Gaussian setting, we also perform simulations using a HMM process to model a setting where non-null where each hypothesis not independent of neighboring hypotheses. We set \(\mu_i = \mu\) to be constant for all \(i\), but we define \(\pi_i\) as follows:
\begin{align*}
	\pi_1 = 0.5, \quad \quad 
	\pi_i = \begin{cases}
		\transprob & \text{if }H_{i - 1} \in \hypset_0,\\
		1 - \transprob & \text{if }H_{i - 1}\not\in \hypset_0.
	\end{cases}
\end{align*}
If we let \(\beta_i\) be the state at the \(i\)th time step, we can view the data generating process as a two state HMM model with a \(\transprob\) probability of transitioning to the other state at the next time step. In expectation, there is an equal number of hypotheses that are null and non-null, regardless of the choice of \(\transprob\). However, a smaller \(\transprob\) will result in large clusters of null or non-null hypotheses in consecutive time steps.

\subsection{Results on HMM}

We show extra results of the \(\SupLORD\) comparison with \(\LORD\) and \(\LORD\FDX\) on the HMM model in \Cref{fig:BaselineHMM}.

\begin{figure}[h!]
	\centering
	\subfigure{
		\includegraphics[width=0.7\textwidth]{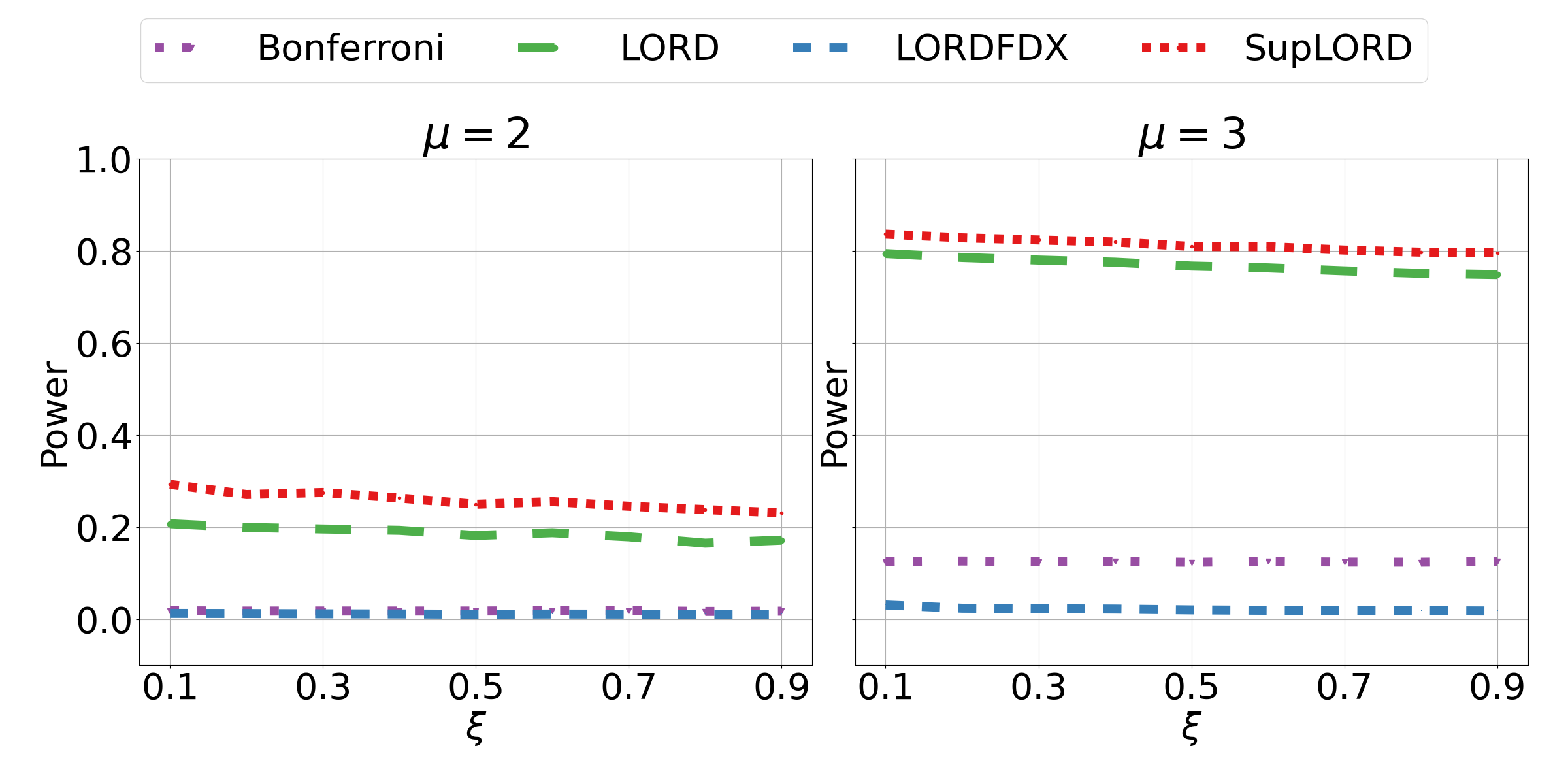}
	}
	
	\subfigure{
		\centering
		\adjincludegraphics[trim={0 0 0 {0.11\height}}, clip,width=0.7\textwidth]{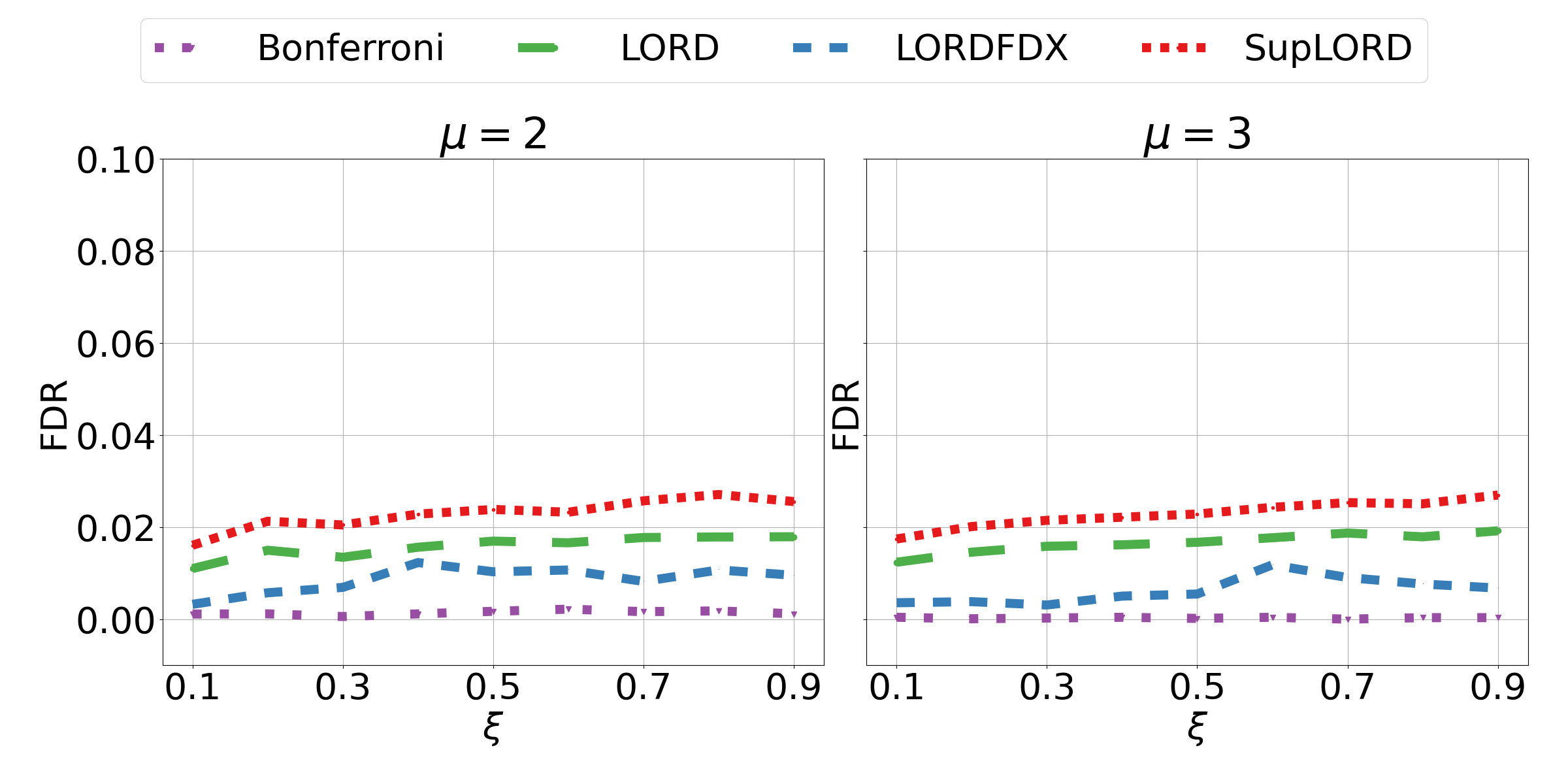}
	}
	
	\subfigure{
		\centering
		\adjincludegraphics[trim={0 0 0 {0.11\height}}, clip,width=0.7\textwidth]{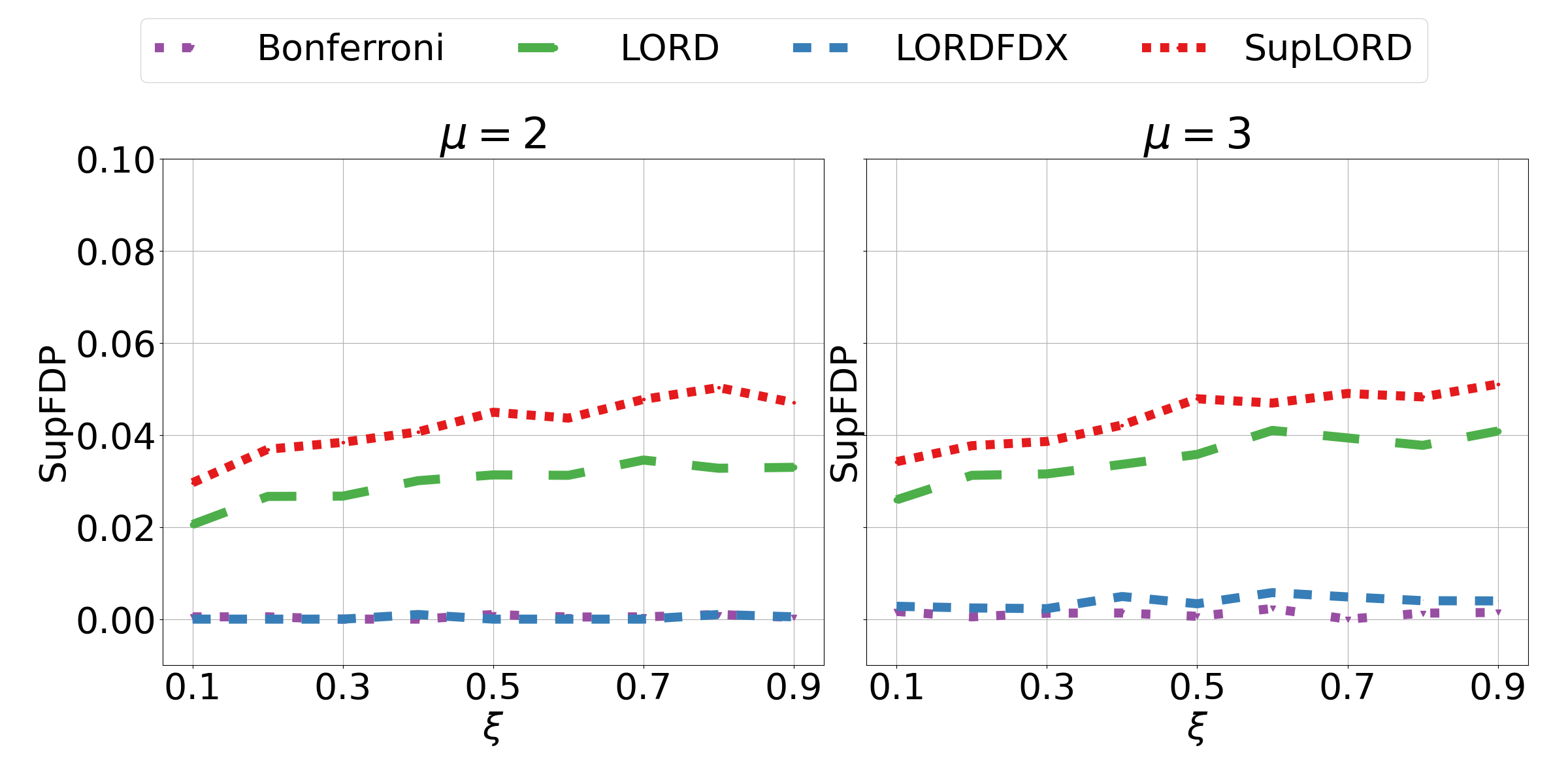}
	}
	
	\caption{Plot of state transition probability, \(\transprob\), vs. FDR, \(\supFD_{30}\), and power where \(\ell=0.05\) for signal strengths of \(\mu \in \left\{2, 3\right\}\) in the HMM data setting. \(\SupLORD\) and \(\LORD\FDX\) are calibrated with \(\delta^*=0.05\) for \(\epsilon^*=0.15\). We observe that \(\SupLORD\) has higher power across all \(\transprob\) and \(\mu\), showing that it is not penalized by clustering of null or non-null hypotheses, relative to prior methods. (Recall that the definition of power only considers correctly rejected hypotheses. The gain in power is due to a less conservative algorithm that makes better use of its \(\FDR\)/\(\FDX\) budget.)}
	\label{fig:BaselineHMM}
\end{figure}

\begin{figure}
	\centering
	\subfigure{
		\includegraphics[width=0.7\textwidth]{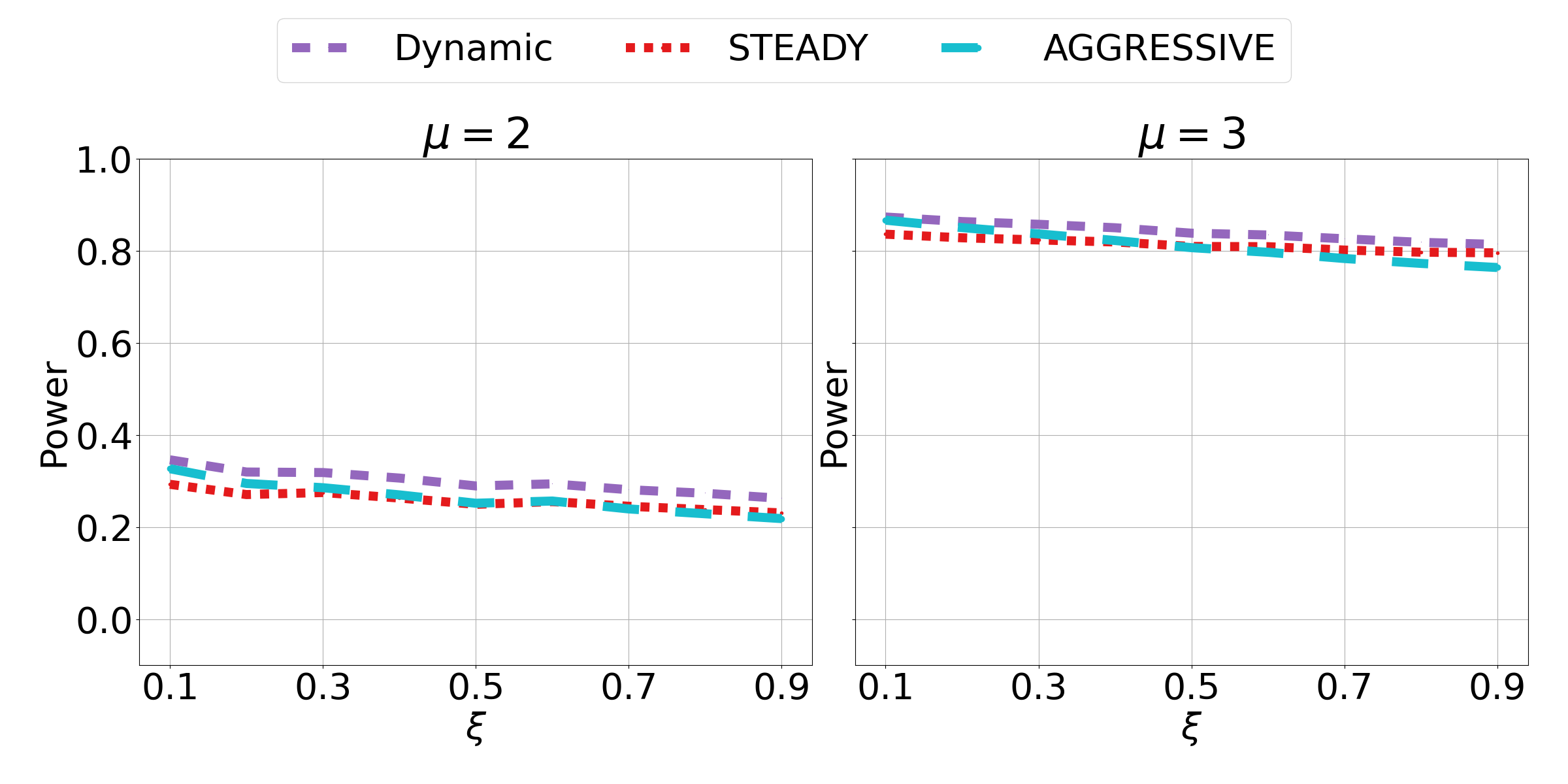}
	}
	
	\subfigure{
		\centering
		\adjincludegraphics[trim={0 0 0 {0.11\height}}, clip,width=0.7\textwidth]{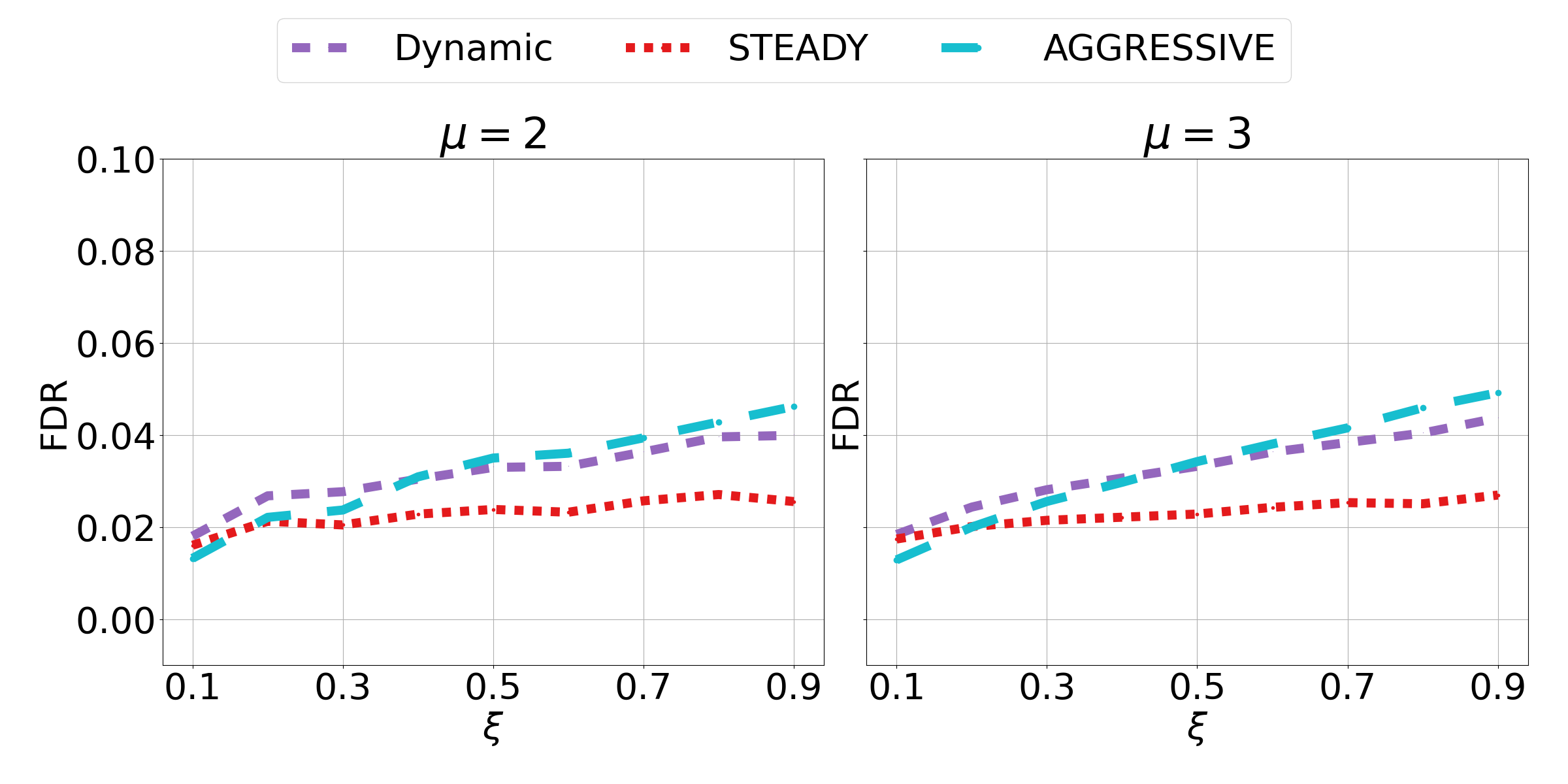}
	}
	
	\subfigure{
		\centering
		\adjincludegraphics[trim={0 0 0 {0.11\height}}, clip,width=0.7\textwidth]{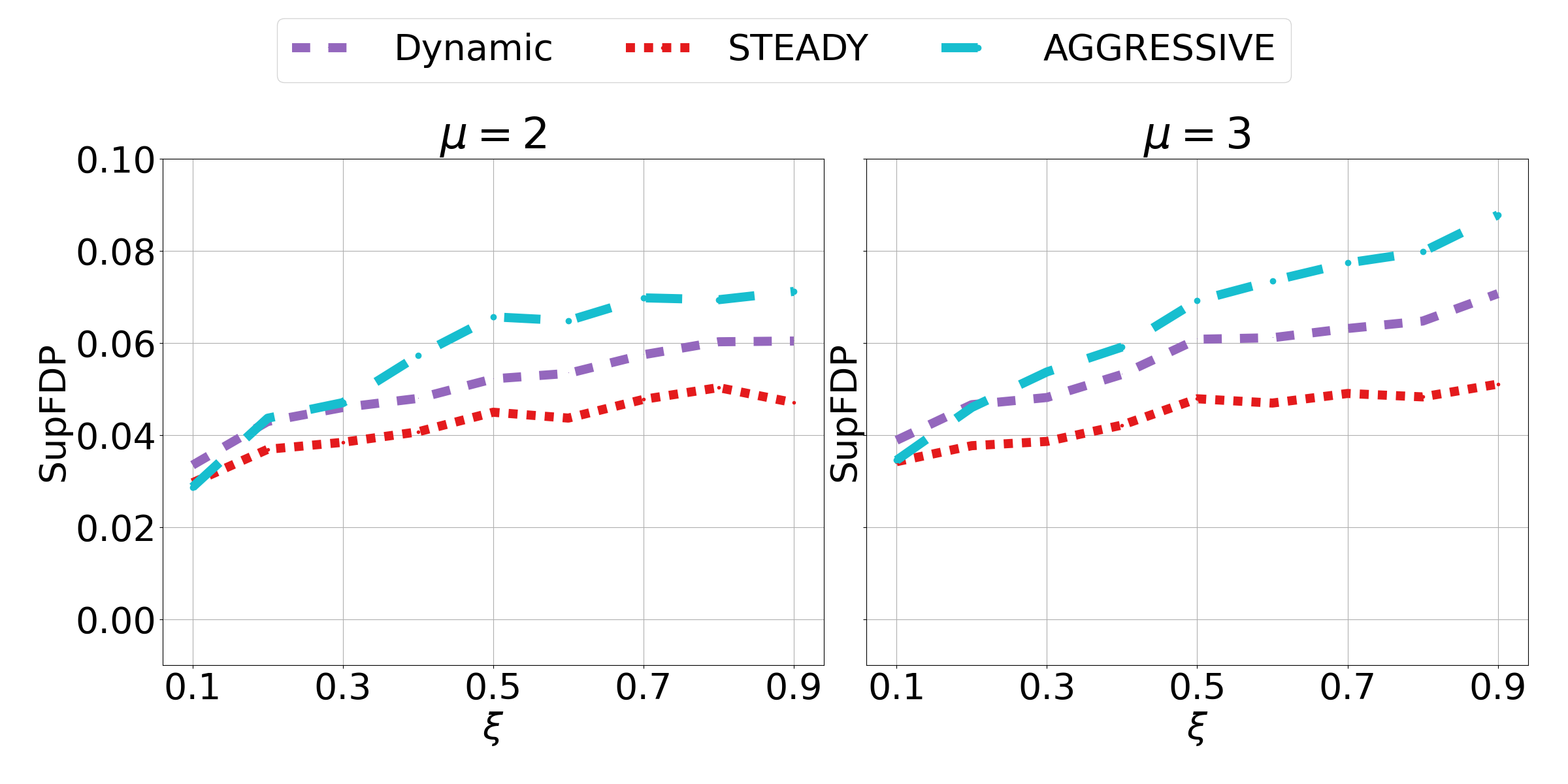}
	}
	\caption{Plot of the state change probability, \(\transprob\), vs. \(\FDR\), \(\supFD_{30}\), and power in the HMM data setting comparing between \(\SupLORD\) with dynamic scheduling and without. The power plots suggests that dynamic scheduling is uniformly better than either \steady\ or \aggressive.}
	\label{fig:DynamicHMM}
\end{figure}

In \Cref{fig:DynamicHMM}, we observe that in the HMM setting, dynamic scheduling remains the best choice, as it is superior or equal to the other schedules in power on all signal strengths and transition probabilities. In addition, if we compare the performance of dynamic scheduling and \aggressive\ we note that the gap between the power of these two schedules is virtually nonexistent when the probability of changing states, \(\transprob\), is low. The exact reverse is true when comparing dynamic scheduling and \steady\ --- they have similar power when \(\transprob\) is high. This reflects the ability of dynamic scheduling to adapt to the distribution of non-null hypotheses by dynamically adjusting its behavior in accordance with the wealth of the algorithm. Consequently, it incorporates the strengths of both \steady\ and \aggressive, and is more powerful regardless of \(\transprob\) as a result.
\clearpage

\revision{\section{Illustrating the Dynamic Scheduling Improvement}
\label{sec:DynamicSchedulingExp}

To illustrate the difference between dynamic scheduling and baseline static schedules, we plot the different alpha values produced by each scheduling algorithm with \(\SupLORD\) for a single trial. In \Cref{fig:AlphaComp}, we plot the difference between alpha values of dynamic scheduling and each of the baseline schedules for a single trial in the constant data setting, where non-null frequency is \(\pi=0.3\) and non-null mean is \(\mu=3\). We see that for the majority of hypotheses, dynamic scheduling outputs \(\alpha_i\) that are larger by about 0.005. Note that a dynamic schedule only makes more rejections if the p-value lies in the gap between the alpha value output by the baseline schedule, and the alpha value output by the dynamic schedule. Further, a p-value from a non-null hypothesis generally lies in this alpha value gap with slightly higher probability than a uniform distribution for the null hypothesis e.g.\ it may be within the gap with 0.01--0.02 probability. Thus, dynamic scheduling gets consistent power gains over each of the baselines by having generally larger alpha values than baseline schedules.} 

\begin{figure}[!h]
    \centering
    \subfigure[Dynamic \(\alpha\) -- \steady\ \(\alpha\)]{
    \centering
    \includegraphics[width=0.6\textwidth]{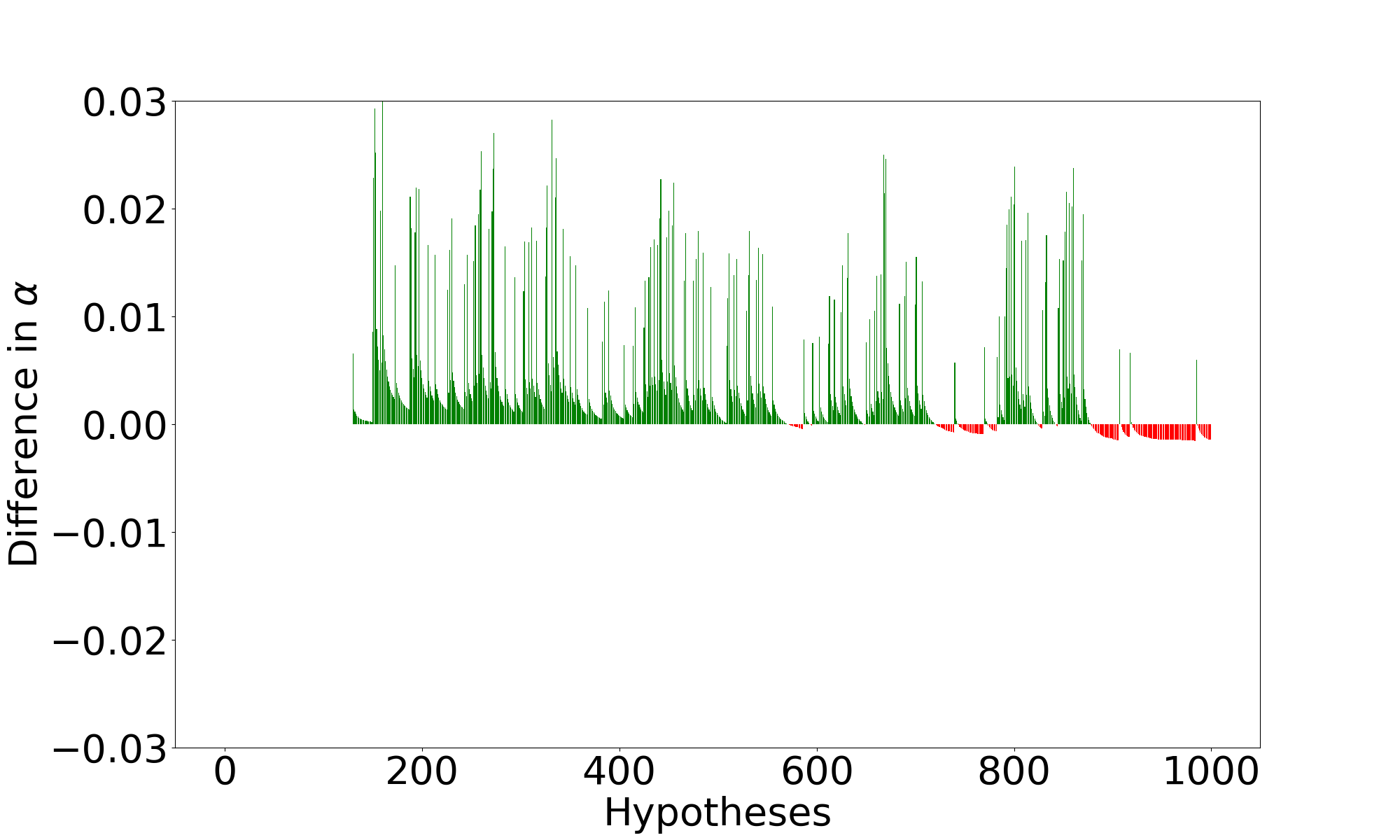}
    }
    
    \subfigure[Dynamic \(\alpha\) -- \aggressive\ \(\alpha\)]{
    \centering
    \includegraphics[width=0.6\textwidth]{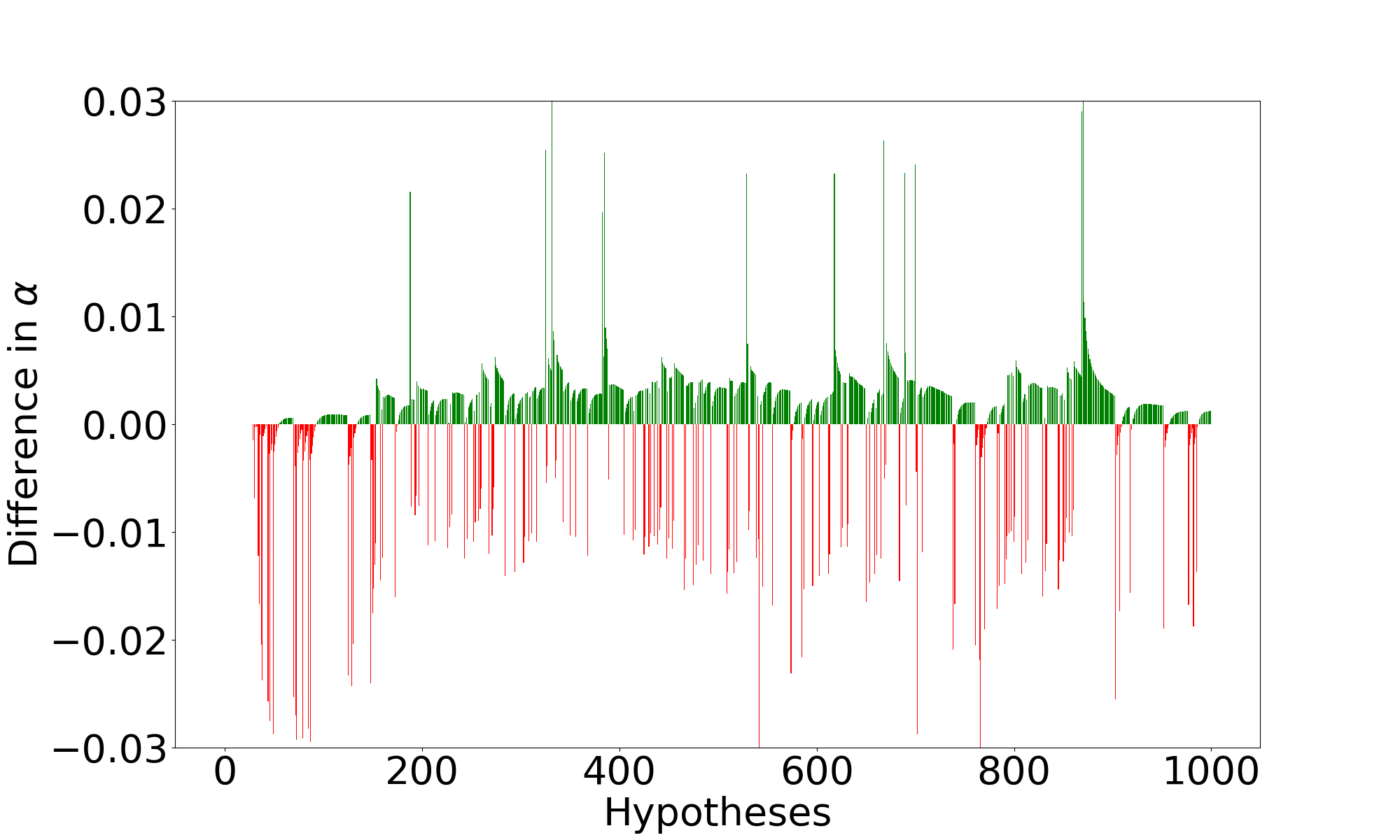}
    }
    \caption{\revision{Plot of differences in alpha values between dynamic scheduling and each baseline static schedule for an identical trial of p-values generated by the same seed in the constant data setting with non-null frequency \(\pi=0.3\), and non-null mean \(\mu=3\). Green bars mark steps where dynamic scheduling has higher alpha value, and red marks steps where the baseline static schedule has a higher alpha value. Dynamic scheduling has higher alpha values by about 0.005 for the majority of the trial when compared to both baseline schedules.}}
    \label{fig:AlphaComp}
\end{figure}

\revision{We plot the actual alpha values for each schedule of the single run in \Cref{fig:SingleAlphaRun} and the mean alpha values of each schedule across 200 runs in \Cref{fig:MeanAlphaRun}. In \Cref{fig:MeanAlphaRun}, we can observe that dynamic scheduling has consistently larger alpha values than both schedules, albeit by a small amount.}
\begin{figure}[h]
    \centering
    \includegraphics[width=0.7\textwidth]{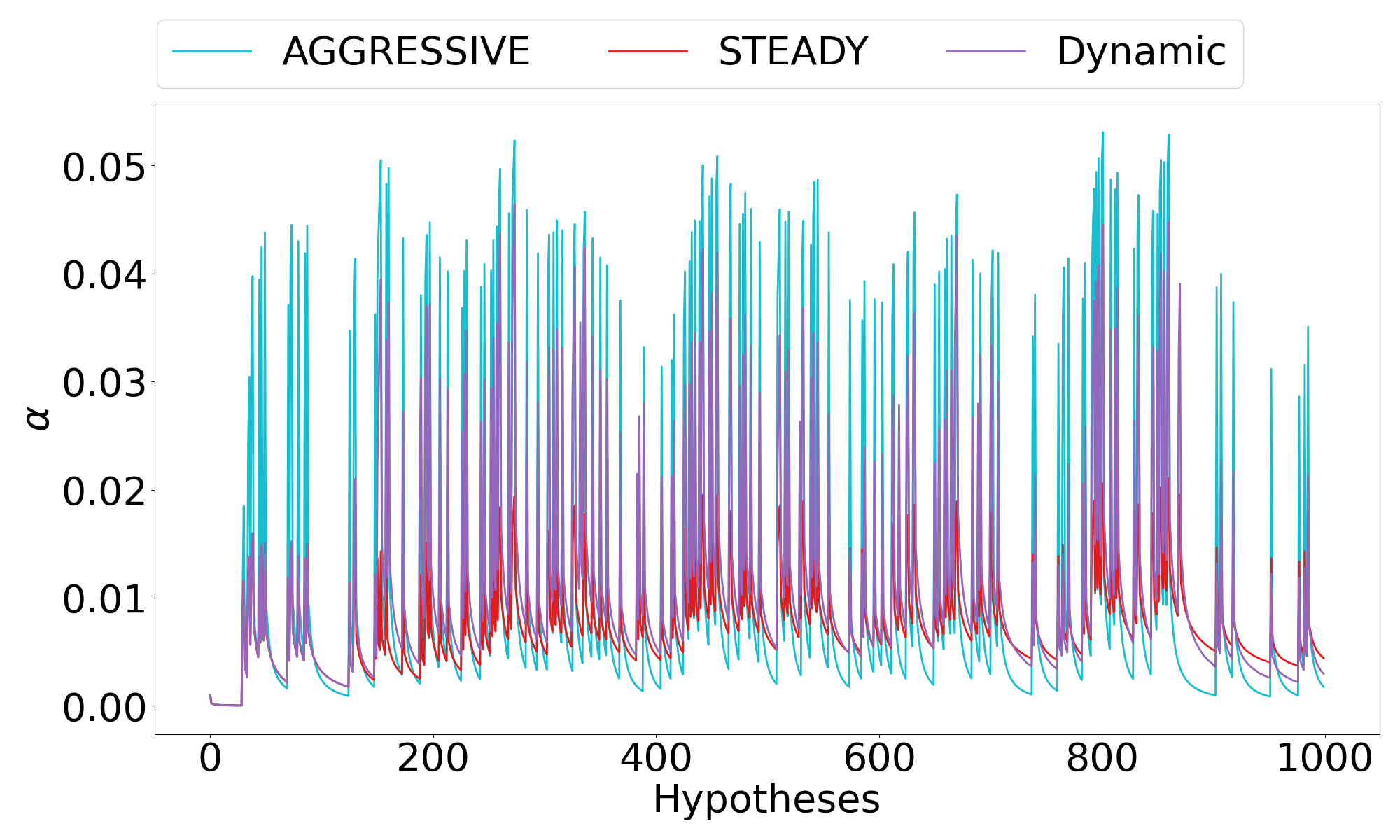}
    \caption{\revision{(Raw values) Plot of actual alpha values between dynamic scheduling and each baseline static schedule for an identical trial of p-values generated by the same seed in the constant data setting with non-null frequency \(\pi=0.3\), and non-null mean \(\mu=3\).}}
    \label{fig:SingleAlphaRun}
\end{figure}
\begin{figure}[h]
    \centering
    \includegraphics[width=0.7\textwidth]{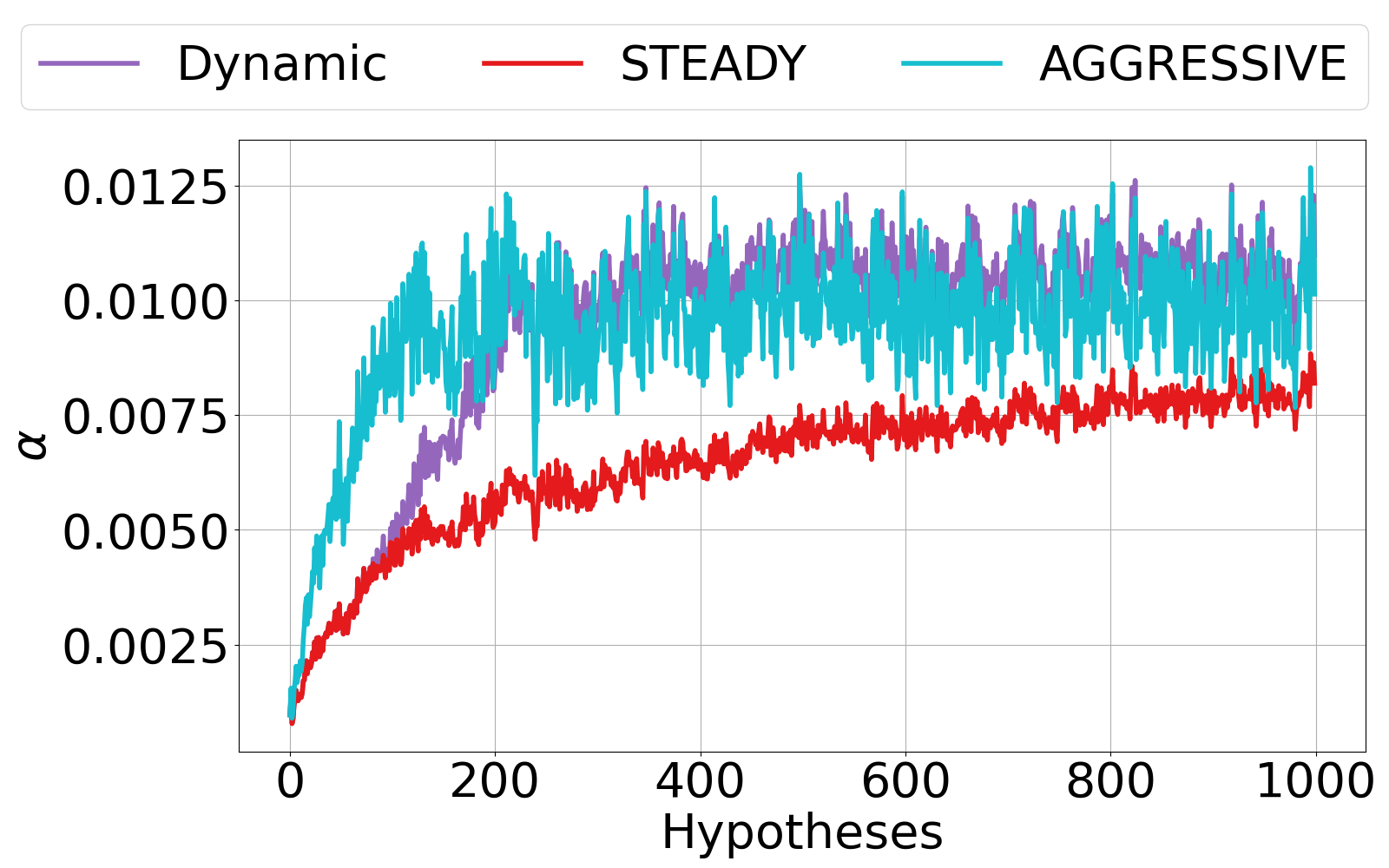}
    \caption{(Mean values) \revision{Plot of mean alpha values between dynamic scheduling and each baseline static schedule in the constant data setting with non-null frequency \(\pi=0.3\), and non-null mean \(\mu=3\) over 200 trials. We can observe that dynamic scheduling is generally larger than both baseline schedules by a small amount in a consistent fashion.}}
    \label{fig:MeanAlphaRun}
\end{figure}

	\vfill
\end{document}